\documentclass[a4paper,UKenglish,cleveref, autoref, thm-restate]{lipics-v2021}
\pdfoutput=1 %uncomment to ensure pdflatex processing (mandatatory \eg to submit to arXiv)
\usepackage{caption}
\nolinenumbers
\usepackage{amssymb}
\usepackage{amsthm}

\usepackage{amsfonts}
\usepackage{stmaryrd}
\usepackage{xcolor, soul}

\usepackage{graphicx}
\usepackage{bbding}
\usepackage{pifont}
\usepackage{paralist}
\usepackage{enumitem}
\usepackage{lineno}
\usepackage{wrapfig}
\usepackage{xspace}
\usepackage[shortcuts]{extdash}
\usepackage{anyfontsize}
\usepackage{multicol}
\usepackage{blindtext}

\DeclareMathAlphabet{\mymathbb}{U}{bbold}{m}{n}

\usepackage{tikz}
\usetikzlibrary{arrows,positioning,shapes,decorations,automata,backgrounds,petri,fit,calc,shapes.multipart,decorations.text,calc,arrows.meta}
\usepackage{relsize}

\newcounter{Reqs}
\Roman{Reqs}

\theoremstyle{claimstyle}
\newtheorem{requirement}[theorem]{Requirement}%[Reqs]

\theoremstyle{definition}
\newtheorem{myclaim}[theorem]{Claim}

\counterwithin{theorem}{section}
\counterwithin{lemma}{section}
\counterwithin{corollary}{section}
\counterwithin{proposition}{section}
\counterwithin{exercise}{section}
\counterwithin{definition}{section}
\counterwithin{conjecture}{section}
\counterwithin{observation}{section}
\counterwithin{example}{section}
\counterwithin{remark}{section}
\counterwithin{note}{section}
\counterwithin{claim}{section}
\counterwithin{figure}{section}
\numberwithin{requirement}{section}
\numberwithin{myclaim}{section}

\AtBeginDocument{%

}

\tikzset{fontscale/.style = {font=\relsize{#1}}}

\newcommand{\grayit}[1]{\textcolor{gray}{#1}}

\newcommand{\blueit}[1]{\textcolor{blue}{#1}}

\newcommand{\replace}[2]{\grayit{\st{#1}}\blueit{#2}}

\newcommand{\add}[1]{\textcolor{blue}{#1}}

\newcommand{\commentout}[1]{}

\newcommand{\dfl}[2]{}%{\coloredfootnote{yellow!70!white}{#1}{#2}}
\newcommand{\esl}[2]{}%{\coloredfootnote{magenta!70!white}{#1}{#2}}

\newcommand{\aggregator}[1]{\ensuremath{\textsf{#1}\xspace}}
\renewcommand{\limsup}{\aggregator{LimSup}}
\renewcommand{\liminf}{\aggregator{LimInf}}
\newcommand{\limsupavg}{\aggregator{LimSupAvg}}

\newcommand{\discsum}{\aggregator{DiscSum}}

\newcommand{\ned}{{\textsc{ned}}\xspace}
\newcommand{\omeganed}{\supomeganed}%
\newcommand{\supomeganed}{\overline{\ensuremath{\omega}}\text{-}\textsc{ned}}

\newcommand{\zero}{\mymathbb{0}}

\newcommand{\true}{\textsc{t}}
\newcommand{\false}{\textsc{f}}

\newcommand{\DM}{\mathbb{DM}}
\newcommand{\DMpk}[1]{\ensuremath{\DM_{{#1}}}^{+}}
\newcommand{\DMnk}[1]{\ensuremath{\DM_{{#1}}^{-}}}
\newcommand{\DMpmk}[1]{\ensuremath{\DM_{{#1}}^{\pm}}}

\newcommand{\wordrank}{\ensuremath{\Omega^{\textsl{word}}}}
\newcommand{\infixrank}{\ensuremath{\Omega^{\textsl{infix}}}}
\newcommand{\letterrank}{\ensuremath{\Omega^{\textsl{letter}}}}

\newcommand{\reset}{\ensuremath{\textsl{r}}}
\newcommand{\domind}{\ensuremath{\textsl{dom-ind}}}
\newcommand{\predomind}{\ensuremath{\textsl{predom-ind}}}

\newcommand{\detind}{\ensuremath{\textsl{det-ind}}}

\newcommand{\predomsuf}{\ensuremath{\textsl{predom-suf}}}

\newcommand{\domsuf}{\ensuremath{\textsl{dom-suf}}}
\newcommand{\detsuf}{\ensuremath{\textsl{det-suf}}}

\newcommand{\inj}{\ensuremath{\textsl{inj}}}

\newcommand{\LTL}{LTL}
\newcommand{\rLTL}{rLTL}

\newcommand{\ltlG}{\ensuremath{\textbf{G}}}
\newcommand{\ltlF}{\ensuremath{\textbf{F}}}
\newcommand{\ltlU}{\ensuremath{\textbf{U}}}
\newcommand{\ltlX}{\ensuremath{\textbf{X}}}
\newcommand{\ltlW}{\ensuremath{\textbf{W}}}

\newcommand{\signal}[1]{\textsl{#1}}

\newcommand{\score}{\ensuremath{\textsl{score}}}

\newcommand{\avgscore}{\ensuremath{\textsl{avgscore}}}
\newcommand{\valrbst}{\ensuremath{\textsl{rbst}}}

\newcommand{\rbstdecomp}{\ensuremath{\textsl{rbst-dcmp}}}

\newcommand{\lcolor}{\ensuremath{\textsl{color}}}
\newcommand{\ecolor}{\ensuremath{\textsl{e-color}}}
 
\newcommand{\icolors}{\ensuremath{\textsl{colors}}}

\newcommand{\leqrbst}{\mathrel{\unlhd}}

\newcommand{\geqrbst}{\mathrel{\unrhd}}
\newcommand{\gtrbst}{\mathrel{\rhd}}

\newcommand{\sema}[1]{\ensuremath{\llbracket#1\rrbracket}}
\newcommand{\aut}[1]{\mathcal{#1}}

\newcommand{\truerobustness}{\text{\FourStar}}
\newcommand{\vigor}{\text{\FourStarOpen}}
\newcommand{\term}[1]{\textcolor{brown}{\textsl{#1}}}
\newcommand{\mathterm}[1]{\textcolor{brown}{{#1}}}

\newcommand{\col}[1]{\ensuremath{\textsl{#1}}}

\newcommand{\alphabetfont}[1]{\textsc{#1}}
\newcommand{\green}{\alphabetfont{g}}	
\newcommand{\red}{\alphabetfont{r}}		
\newcommand{\yellow}{\alphabetfont{y}}
\newcommand{\black}{\alphabetfont{b}}
\newcommand{\white}{\alphabetfont{w}}

\newcommand{\fg}{\ensuremath{\textsl{fg}}}

\newcommand{\ulcolor}[2]{{\setulcolor{#1}{\setul{.6pt}{2pt}\ul{#2}}}}

\newcommand{\rdl}[1]
{\ulcolor{red}{#1}}
\newcommand{\yll}[1]
{\ulcolor{yellow}{#1}}
\newcommand{\grl}[1]
{\ulcolor{green}{#1}}
\newcommand{\bll}[1]
{\ulcolor{black}{#1}}

\tikzset{pyellow/.style={preaction={ 
draw,yellow,-, 
double=yellow,
line width=1.1pt,
}}}
\tikzset{pgreen/.style={preaction={ 
draw,green!60,-, 
line width=1.1pt,
double=green!60,
}}}
\tikzset{pred/.style={preaction={ 
draw,magenta!60,-, 
double=magenta!60,
line width=1.1pt,
}}}

\tikzset{pblack/.style={preaction={ 
draw,gray!80,-, 
double=gray!80,
line width=1.1pt,
}}}

\tikzset{pwhite/.style={preaction={ 
draw,gray!20,-, 
double=gray!20,
line width=1.1pt,
}}}

\tikzset{pblue/.style={preaction={ 
draw,blue!40,-, 
double=blue!40,
line width=1.3pt,
}}}

\tikzset{porange/.style={preaction={ 
draw,orange!50,-, 
double=orange!50,
line width=1.3pt,
}}}

\tikzset{pbrown/.style={preaction={ 
draw,brown!50,-, 
double=brown!50,
line width=1.3pt,
}}}

\tikzset{pteal/.style={preaction={ 
draw,teal!50,-, 
double=teal!50,
line width=1.3pt,
}}}

\tikzset{
    side by side/.style 2 args={
    line width=5pt, -,
    #2, 
    postaction={
        clip,postaction={draw,#1}
        }
    }
}

\title{$\omega$-Regular Robustness}

\author{Dana Fisman}{Department of Computer Science, Ben-Gurion University, Israel }{dana@bgu.ac.il}{https://orcid.org/0000-0002-6015-4170}{}

\author{Elina Sudit}{Department of Computer Science, Ben-Gurion University, Israel }{elinasu@post.bgu.ac.il}{https://orcid.org/0009-0009-6187-6894}{}

\authorrunning{D. Fisman and E. Sudit} 

\begin{document}

\maketitle

\begin{abstract}
Roughly speaking, a system is said to be \emph{robust} if it can resist disturbances and still function correctly. For instance, if the requirement is that the temperature remains in an allowed range $[l,h]$, then a system that remains in a range $[l',h']\subset[l,h]$ is more robust than one that reaches $l$ and $h$ from time to time. In this example the initial specification is quantitative in nature, this is not the case in $\omega$-regular properties.
Still, it seems there is a natural robustness preference relation induced by an $\omega$-regular property. E.g. for a property requiring that every request is eventually granted, one would say that a system where requests are granted two ticks after they are issued is more robust than one in which requests are answered ninety ticks after they are issued.
 In this work we manage to distill a \emph{robustness preference relation} that is induced by a given $\omega$-regular language. 
The robustness preference relation is a semantic notion (agnostic to the given representation of the language) that relies on Wagner's hierarchy and on Ehlers and Schewe's definition of natural rank of infinite words. 
 It aligns with our intuitions on common examples, satisfies some natural mathematical criteria, and
 refines Tabuada and Neider's five-valued semantics 
 into an infinite domain.
\end{abstract}

\section{Introduction}\label{sec:intro}

A temporal logic formula classifies execution paths into two: the ones that satisfy  it and the ones that do not.
In various settings a finer classification is needed.
Various quantitative formalisms can be used to specify which among two satisfying (resp. violating) executions {is} better (resp. worse). 
To qualify \emph{better}, a natural notion is that of \emph{robustness}.
Loosely speaking, a system is \emph{robust} if it can resist more disturbances and faults and still satisfy the property.
For instance, consider a property requiring a surface to endure a weight of 60 Kg.
Among two systems that satisfy it the one that endures a larger weight is more robust.
Among two systems that violate it, still the one that endures a larger weight is preferable.
In this example the nature of robustness is tied to the quantitative nature of weight, a similar quantitative aspect is present in many other studied notions of robustness.
We, on the other hand, are interested in the robustness preference relation that seems to be induced from an $\omega$-regular property itself, e.g. from a plain temporal logic formula in, say, LTL.

\vspace{-3mm}
\subparagraph*{Notions of robustness}
Various works have considered problems related to robust verification and synthesis.
In particular, robustness issues have been tackled in the context of timed systems, hybrid systems and cyber-physical systems~\cite{HenzingerR00,DeWulftDMR08,JaubertR11,BouyerMS15,FinkbeinerFKK22}.
 The nature of robustness studied in these works relates to the inherent quantitative nature of these systems, e.g. realtime values of the atomic signals and/or the continuous (non-discrete) nature of time, and is thus very different from the notion of robustness we are looking for. 
 
 Focusing on $\omega$-regular languages in the Boolean (aka qualitative) setting, numerous works
 considered robustness from the angle of assume-guarantee \cite{KupfermanV97,Ehlers11,BloemGHKK12,BloemCGHHJKK14,BloemCES22}.
 Considering robustness from a view that relates to cost of errors and edit distance metrics was pursued in~\cite{FiliotMRST20} 
 which suggested using weighted transducers for modeling deviations and faults by which one would like to measure the robustness of systems.  Working with weighted transducers provides a general way to model deviations, but it leaves the task of defining how to model deviations to an outside source.
 These and other works on quantitative reasoning assume the automata or temporal logic are equipped with some numeric values reflecting  weights or costs~\cite{ChechikED01,DrosteG07,KL07,FaellaLS08,ChatterjeeDH10,AlurFT12,AlurDDRY13,BokerCHK14,BouyerGM14,AlmagorBK16,Kupferman22}. We are interested in a way to seamlessly obtain a robustness measure from a typical $\omega$-regular specification (e.g. given by a temporal logic formula or an $\omega$-automaton).

\vspace{-3mm}
\subparagraph*{The desired notion of robustness}
It seems there is a natural notion of robustness that can be inferred from an $\omega$-regular property 
without additional means. We find that in many properties recurring in the literature, when comparing two words, it is obvious which should be considered more robust wrt the specification at hand. For example, in properties capturing response time one always considers words where the waiting time is shorter as more robust. Our main goal in the work presented here is to manage to distill the natural robustness preference relation that an $\omega$-regular language induces. 

A great inspiration for us is the work of Tabuada and Neider~\cite{TabuadaN16} which proposed \rLTL, a multi-valued refinement of the semantics of the temporal logic \LTL.
Indeed, in terms of syntax the logic \rLTL\ is identical to \LTL; the difference is in the semantics. That is,
in contrast to the above mentioned works, rLTL captures robustness that relates solely to a plain (non-quantified) LTL specification, and does not require additional inputs from the user. 
 As a motivating example Tabuada and Neider give the following example.

\begin{example}[Basic TN]\label{ex:TN}
    Consider the formula $\ltlG a$ (always $a$) over the alphabet $\Sigma=\{a,b\}$.
    The word $w_1=a^\omega$ satisfies it, and the words  $w_2=bbba^\omega$, $w_3=(ab)^\omega$, $w_4=aaab^\omega$, $w_5=b^\omega$ do not. Still, our intuition tells us that $w_2$ falsifies it less than $w_3$, that $w_3$ falsifies it less than $w_4$, and that $w_4$ less than $w_5$.
Indeed, comparing $w_2$ and $w_3$, we note that $w_2$ satisfies it all along except in the first three ticks whereas in $w_3$ it is violated infinitely often.
Similar intuitions entail that $w_3$ falsifies it less than $w_4$ since in the former infinitely often the desired signal $a$ holds, whereas in the latter from some point onward $a$ no longer holds. 
Similar reasoning warrant that  $w_4$ falsifies it less than $w_5$, in which it is constantly violated. 
\end{example}

The formal explanation of
 these intuitions that 
$w_i$ is more robust (wrt $\ltlG a$) than $w_{i+1}$ for $1\leq i < 5$ considers the formulas $\varphi_1=\ltlG a$ (always $a$), $\varphi_2=\ltlF\ltlG a$ (eventually always $a$), $\varphi_3=\ltlG\ltlF a$ (always eventually  $a$), $\varphi_4=\ltlF a$ (eventually  $a$), $\varphi_5=\ltlG \neg a$ (never  $a$). 
Comparing the words $w_1,w_2,w_3,w_4,w_5$ we have that $w_i$ satisfies $\varphi_i$ yet $w_i$ does not satisfy $\varphi_j$ for $j<i$. To capture this Tabuada and Neider propose a $5$-valued semantics, where the $5$ values, roughly speaking, correspond to these five representative formulas.
This idea has been extended to various other temporal logics~\cite{NayakNRZ22,NeiderWZ22,MuranoNZ23} and is the baseline for the notion we seek for. 

Inspired by these intuitions, we look for a preference relation on words, induced by an $\omega$-regular language $L$ that is refined into an infinite-value semantics, in the following way. 
\begin{example}[Refined TN]\label{ex:refine-TN}
    Using rLTL we can distinguish the satisfaction of 
    $w_1,...,w_5$ above wrt $L_{\ltlG a} = \ltlG a$.
    Since 
indeed $w_i\models \varphi_i$ and $w_i \not\models  \varphi_j$ for $j<i$.
However, we cannot distinguish 
$w_3'=(a^{1000}b a^{999})^\omega$ from $w''_3=(aba)^\omega$ since both satisfy $\varphi_3$ and do not satisfy $\varphi_1$ or $\varphi_2$.
But our intuition says that $w'_3$ is more robust than $w''_3$ since  only one change in $2000$ ticks is required to make $w'_3$ satisfy $\ltlG a$ whereas one change in every $3$ ticks is required to make $w''_3$ satisfy $\ltlG a$.
Note that Tabuada \& Neider's approach regards the non-periodic, aka \emph{transient} part of the word and deems $w_2=bbba^\omega$ less robust than $w_1=a^\omega$. We would like to make this distinction more refined and  
comparing 
$w'_2=b^{2}a^\omega$ and $w''_2=b^{100}a^\omega$ declare that  $w'_2$ is more robust than $w''_2$ (wrt $L_{\ltlG a}$).
Similarly, we would like to distinguish
$w'_4=a^{10000}b^\omega$ from $w''_4=a^{3}b^\omega$ and say  that (wrt $L_{\ltlG a}$) $w'_4$ is more robust than $w''_4$. 
\end{example}

\subparagraph*{Distilling the robustness preference relation}
Our aim is to develop a notion of robustness that matches the above intuitions (and some natural mathematical criteria) and is derived \emph{solely} from a given $\omega$-regular language $L$. 
The idea is to \emph{distill} from $L$ itself a value function $\valrbst_L$ that given a lasso word returns  a value that captures  
its \emph{robustness value} wrt $L$.
The values are 
triples where the first component simply reflects the acceptance of the word in the language, and the second and third components provide the robustness of the periodic and transient parts of the word. 
Note that the order matches the significance: obviously, acceptance is the most significant, and the robustness value of the period outweighs the robustness value of the spoke (the transient part). 
But what is the robustness value of a period or spoke of a lasso word?
Below we try to give some intuition, focusing on completely periodic words to simplify the matter. 

In the languages $L_{\ltlG a}$ and $L_{\infty a}$ (requiring infinitely many $a$'s) periodic words  with more $a$'s in the period are more robust. How can we lift this intuition to the general case? 
Consider the language $\infty aa$ requiring infinitely many occurrences of $aa$ over $\Sigma=\{a,b,c\}$. Words with many occurrences of $aa$ are more robust than those with few. Accordingly, the word $(caa)^\omega=(caacaa)^\omega$ is  more robust than $(cccaac)^\omega$ since in the first $aa$ occurs twice in a period of 6 letters, and in the second only once. 
In the language $\neg \infty bb$ requiring $bb$ occurs finitely often, words with fewer occurrences of $bb$ are more robust. In the language $L=\infty aa \wedge \neg\infty bb$ occurrences of $aa$ increase the robustness, and occurrences of $bb$ decrease the robustness. Thus, loosely speaking, we would like to somehow associate with each letter a grade reflecting how it contributes to the robustness of the word (in a good way, in a bad way, or neutrally). 
Using green (resp. red) underlines to mark the letters that increase (resp. decrease) the robustness of the word we get the following wrt $L$: 
$(ccca$\grl{$a$}$)^\omega$,  $(cccb$\rdl{$b$}$)^\omega$, and $(ca$\grl{$a$}$ca$\grl{$a$}$cb$\rdl{$b$}$)^\omega$.  
Once each letter has a color reflecting whether it increases or decreases the robustness of the word, the robustness value of the period can somehow aggregate and average these values. For instance, it can be $\frac{1}{l}(\textsc{g}-\textsc{r})$ where $l$ is the length of the period, and $\textsc{g}$ and $\textsc{r}$ are the number of green and red letters in the period, resp. The actual definition is more involved, but this is the rough idea.

The challenge is to find a semantic way to color individual letters of a word (wrt a language). The semantic definitions borrow intuition from the Wagner hierarchy~\cite{Wagner75} and rely
on the notion of natural rank of an infinite word~\cite{EhlersS22}.\footnote{The term used in \cite{EhlersS22} is \emph{natural colors} but we want to reserve the word \emph{color} for the colors (green, red, etc.) of letters as mentioned above.} The latter notion has been extended to provide a natural rank for an infix of a period of an infinite word~\cite{BohnL24}. We start by providing a further extension that provides ranks for individual letters of an infinite word. The color of the $i$-th letter is then derived from the ranks of the $(i{-}1)$-th and $i$-th letters.

To define letter ranks we introduce the semantic notion of \emph{dominant-suffixes}. 
From the latter we entail a definition of equivalence classes $\equiv^\vigor$ and $\equiv^\truerobustness$ that yield a definition of what we denote $\aut{P}^\truerobustness_L$ and term \emph{the robustness parity automaton} for $L$. 
The size of $\aut{P}^\truerobustness_L$ may be exponential in the size of the minimal parity automaton for $L$. This is expected since it conveys more information, and all state-of-the-art canonical representations for $\omega$-regular languages have the same lower bound blowup.\footnote{The canonical representations known so far for $\omega$-regular languages are: the syntactic FORC~\cite{MalerS97}, the recurrent FDFA~\cite{AngluinF16}, the limit FDFA~\cite{LST23}, the colorful FDFA~\cite{FismanGZ24} and the precise parity automaton~\cite{BohnL24}.}

Due to space restrictions, some proofs, examples and figures are deferred to the appendix.

\vspace{-2mm}
\section{Preliminaries}\label{sec:prelim}
\vspace{-1mm}
We use $\mathbb{B}$ for the set $\{\true,\false\}$ or $\{0,1\}$. 
For $i\leq j$, we use $[i..j]$ to denote the set $\{i,i{+}1,...,j\}$.
Given a word $v=\sigma_1\sigma_2...\sigma_m$ we use $v[i..j]$ for the infix of $v$ starting at $\sigma_i$ and ending in $\sigma_j$, inclusive;
and $v[..i]$ (resp. $v[i..]$) for the prefix (resp. suffix) of $v$ that ends (resp. starts) in $\sigma_i$,  inclusive. 
We use parenthesis instead of brackets if we want to exclude $i$ and/or $j$. E.g. $v[2..5)=\sigma_2\sigma_3\sigma_4$ and $v(2..5)=\sigma_3\sigma_4$. 
We use $x\prec y$ to denote that $x$ is a prefix of $y$.

The set of all infinite words over $\Sigma$ is denoted $\Sigma^\omega$. A word $w\in\Sigma^\omega$ is termed a \term{lasso word} if $w=u(v)^\omega$ for some $u\in\Sigma^*$ and $v\in\Sigma^+$.
We refer to $u$ as the \term{spoke} of the lasso, and to $v$ as its \term{period}.
A pair $(u,v)$ for $u\in\Sigma^*$ and $v\in\Sigma^+$ is a \term{representation} of the lasso word $u(v)^\omega$. Note that a lasso word $uv^\omega$ has many representations. E.g. $w=aa(ba)^\omega$ can be represented as $(aa,ba)$, $(a,ab)$, $(aabab,abab)$ etc.

A \term{deterministic $\omega$-automaton} is a tuple $\aut{A}=(\Sigma, Q, q_0, \delta,\alpha)$ consisting of an alphabet $\Sigma$, a finite set $Q$ of
states, an initial state $q_0\in Q$,  a transition function ${\delta: Q \times \Sigma \rightarrow Q}$, and an acceptance condition $\alpha$.
We refer to the first four components of an automaton as the \term{automaton-structure}.
A run of an automaton is defined in the usual manner.
An $\omega$-automaton accepts a word if the run on that word is accepting. There are various $\omega$-acceptance conditions in the literature (e.g. parity, Rabin, Streett, Muller); all are defined with respect to the set of states visited infinitely often during a run. For a run $r=q_0q_1q_2...$ we define $\mathterm{\inf(r)}= \{ q \in Q ~|~ \forall i\!\in\!\mathbb{N}.\ \exists j\!>\!i.\ q_j=q\}$. 

The  $\omega$-acceptance condition we work with is parity.
    A \term{parity} acceptance condition is a mapping $\kappa:Q\rightarrow [i..j]$ of the states to a number in a given bounded set $[i..j]$, which we refer to as \emph{ranks} (and we refer to $\kappa(q)$, for a state $q$, as the \term{rank} of $q$). 
    For a subset $Q'\subseteq Q$, we use $\kappa(Q')$ for the set $\{\kappa(q)~|~q\in Q'\}$. A run  $r$ of a parity automaton is accepting if the minimal rank in $\kappa(\inf(r))$ is even.
We use DPA   for deterministic parity  automata. 
All languages mentioned in the paper are assumed to be $\omega$-regular, i.e. accepted by DPA.

We use $\sema{\aut{A}}$ to denote the set of words accepted by a given automaton $\aut{A}$. For $u\in\Sigma^*$, we use $\aut{A}(u)$ for the state we reach on reading $u$ from the initial state. For $q\in Q$, we use $\aut{A}_q$ for $(\Sigma,Q,q,\delta,\alpha)$, namely the automaton obtained from $\aut{A}$ by making $q$ the initial state.
We associate with $\aut{A}$ its automaton graph $G=(V,E)$
where $V=Q$ and 
$E=\{(v,v')~|~v'\in\delta(v,\sigma)$ for some 
$\sigma \in \Sigma\}$. 
A subset $C$ of vertices (i.e., states) is termed a \term{strongly connected component} (SCC) if there is a non-empty path between any pair of vertices in $C$. 
 An SCC is termed a 
\term{maximal strongly connected component} (MSCC) if there is no SCC 
$C'\supsetneq C$. Note that a run of an automaton $\aut{A}$ on a word $w$ must eventually stay within a single MSCC (i.e. visit no state outside of this MSCC). We term such an MSCC the \term{final} MSCC of $w$ wrt $L$.

The \term{syntactic right congruence relation} for an $\omega$-language $L$ relates two finite words $x$ and $y$ if there is no infinite suffix $z$ differentiating them, that is, for $x,y\in\Sigma^*$, $x\sim_L y$ if $ \forall z\in\Sigma^\omega (xz\in L \iff yz \in L)$. We use  $[u]_{\sim{L}}$ (or simply $[u]$) for the equivalence class of $u$ induced by $\sim_L$.
A right congruence $\sim$ can be naturally associated with an automaton structure $( \Sigma, Q, q_0, \delta )$ as follows: the set of states $Q$ are the equivalence classes of $\sim$. The initial state $q_0$ is the equivalence class $[\epsilon]$. The transition function $\delta$ is defined by $\delta([u],\sigma)=[u\sigma]$.  We use $\mathterm{\aut{A}[{\sim}]}$ to denote the automaton structure induced by $\sim$.

\vspace{-2mm}
\section{Requirements for the robustness preference relation}\label{sec:requirements-from-robustness}
\vspace{-1mm}
Given an $\omega$-regular language $L$, we would like to define a preference relation on $\omega$-words, denoted $\gtrbst_L$, such that $w_1 \gtrbst_L w_2$ conveys that $w_1$ is more robust than $w_2$ with respect to $L$. We can do so by defining a value function $\valrbst_L:\Sigma^\omega\to\mathbb{T}$  and then defining $w \geqrbst_L w'$ if $\valrbst_L(w)\geq \valrbst_L(w')$. 
In the following we try to formalize the intuitions and  requirements for this definition.

\begin{requirement}[Reflect disturbances wrt  $L$]\label{req:respect}
First $\geqrbst_L$ has to \emph{respect} $L$. Meaning, words in $L$ should be preferred over words not in $L$.
Second, when comparing two words $w$ and $w'$ that are both in $L$, roughly speaking, we would like $w$ to be more robust than $w'$ if more changes are required to make $w$ not in $L$ (compared to $w'$). If both  $w$ and $w'$ are not in $L$ we would like $w$ to be more robust if fewer changes are required to make it in $L$. 
For example, with respect to $L_{\infty a}$, both $a^\omega$ and $(ab)^\omega$ are in the language.
    We expect $a^\omega$ to be preferred over $(ab)^\omega$ since intuitively, in the former there are more options to deviate (by taking an action different from $a$) and still satisfy the property.
\end{requirement}

\begin{requirement}[Refine Tabuada and Neider's idea~\cite{TabuadaN16}]\label{req:TN}
We would like the relation to refine the idea behind the  $5$-valued semantics of rLTL~\cite{TabuadaN16} into an infinite domain. 
In particular, following \autoref{ex:refine-TN}, wrt $L_{\ltlG a}$ we require that for every $i<j$ the relation should satisfy
\quad
    $\bullet\ \ b^i(a)^\omega \gtrbst_L b^{j}(a)^\omega$ \quad 
    $\bullet\ \ (b^ia)^\omega \gtrbst_L (b^ja)^\omega$ \quad 
    $\bullet\ \ (ba^j)^\omega \gtrbst_L (ba^i)^\omega$.
\end{requirement}

\begin{requirement}[Convey response time]\label{req:res-time}
Wrt the common LTL property $L=\ltlG (a \to \ltlF b)$ over $\Sigma=\{a,b,c\}$   we would like $(ac^ib)^\omega \gtrbst_{L} (ac^{j}b)^\omega$ for every $i<j$, since the response time in the former is shorter. 
For more intuition on this see \autoref{ex:responset_time} and \autoref{ex:no-redundant-grants} in the appendix.
\end{requirement}

\begin{example}[Infinitely often]\label{ex:infty-ab}
In a broader sense, if there is a requirement that some sequence of events occurs infinitely often, we prefer words where it occurs as often as possible.
Consider $\Sigma=\{a,b\}$ and the language $L_{\infty ab}$ requiring that infix $ab$ occurs infinitely often. The words $(ab)^\omega$ and $(ba)^\omega$ should be the most robust, and more robust than $(aabb)^\omega$, for instance. 
\end{example}

\begin{requirement}[Sensitivity to both the periodic and transient parts of the word]\label{req:period-vs-spoke}
Considering a lasso word 
$u(v)^\omega$, the periodic part $v$ is more significant than $u$ since $v$ occurs infinitely many times in the word whereas $u$ occurs just once. 
For this reason many aggregators considered in analysis of infinite words~\cite{BokerCHK14} 
such as $\limsup$, $\liminf$, $\limsupavg$, etc. essentially ignore the transient part. 
An exception is $\discsum$ (discounted sum) which gives more weight to early events and thus is affected more by the transient part. The latter is definitely not suitable here, since we'd like the periodic part to be more significant. But neither are the formers, since they completely disregard the transient part, 
whereas if two words agree on the period  we would like the spoke to determine which one is more robust.
This is the case for instance in \autoref{ex:refine-TN} where wrt $\ltlG a$ we expect $b^i(a)^\omega \gtrbst_L b^{j}(a)^\omega$ if $i<j$.
\footnote{We note 
that in \cite{FismanGW23} the well known notion of normalized edit distance on finite words ($\ned$) was generalized to a \emph{normalized edit distance on infinite words} ($\omeganed$). 
While being defined as a function between a pair of words (rather than as a relation or value function induced from a given language) it does capture some of the intuitions discussed above, in particular measuring how many letters one should change to transform a word into another one. 
However, since it is based on the notion of $\limsup$ 
it 
completely disregards the spoke of the word, contrary to the classification of Tabuada and Neider, and its refinement sought for here.
}
\end{requirement}

\begin{requirement}[Duality]\label{req:duality}
Robustness is sometimes used when considering words (resp. languages) that satisfy (resp. are subsumed by) the given language $L$ and sometimes when considering words (resp. languages) that falsify (resp. are not subsumed by) the given language $L$. It is important for us that the defined robustness preference relation would refine both accepted and rejected words, and would do so in a symmetrical manner. Ideally, it should satisfy $w_1 \geqrbst_{L} w_2$ iff $w_1 \leqrbst_{\overline{L}} w_2$ where $\overline{L}$ is the complement language of $L$.
\end{requirement}

\begin{requirement}[Satisfaction bar]\label{req:bar-score}
As per Req.~\ref{req:TN}, with respect to $L=L_{\ltlG a}$ we require $w_1{=}(a)^\omega\gtrbst_L w_2{=}bbb(a)^\omega\gtrbst_L w_3{=}(ba)^\omega\gtrbst_L w_4{=}aaa(b)^\omega\gtrbst_L w_5{=}(b)^\omega$. 
What if we consider $L_{\infty a}$ 
instead? We still expect the same order between these words. However, we would like to convey that for $L_{\ltlG a}$, the word $w_1$ is the only one satisfying it, whereas for $L_{\infty a}$ all words $w_i$ for $i\leq 3$ satisfy it. Using $\valrbst_L(w)$, the robustness value associated with word $w$ wrt  $L$, we require there is a clear separation  between values of accepted and rejected words. Specifically, we demand $\valrbst_L(w){\geq} \zero$ iff $w{\in} L$ (where $\zero$ is a special element of the value domain $\mathbb{T}$).
\end{requirement}

\begin{requirement}[Avoid debts]
\label{ex:diff-priorities-for-red-yellow-green-cycles}
Consider the alphabet $\{a,b\}$ and the language $L_{a\text{-seq}}=\infty a \wedge (\infty aa \rightarrow \infty aaa)$ 
requiring infinitely many $a$'s and demanding infinitely many $aaa$'s if there are infinitely many $aa$'s.
We expect 
 $(ab)^\omega \gtrbst_{L_{a\text{-seq}}} (aaab)^\omega$  
since after taking $a$ it is safer to take a $b$ than to take an $a$ which requires taking a third $a$ immediately after; or more accurately, seeing $aa$ infinitely often requires seeing $aaa$ infinitely often.
We expect the most preferred word wrt $L_{a\text{-seq}}$ to be $a^\omega$ since in debt terms we are in debt only once, and there is no word that creates zero debts. 
{A perhaps simpler example for the rationale of avoiding debts is given in \autoref{ex:avoid-debts}.}
\end{requirement}

The challenge in obtaining such a relation is that we require it to be derived solely from the given $\omega$-regular language, without relying on a specific 
representation of it. That is, we want a \emph{semantic} rather than a \emph{syntactic} definition of the robustness preference relation.

\vspace{-4mm}
\section{Towards the definition of robustness preference relation}\label{sec:toward}
\vspace{-1mm}
Next we would like to obtain a formal definition of the robustness value with respect to a given language. We require our definition to be semantic, namely derived from the language itself rather than from a particular representation of the language.
To this aim we  inspect our examples with respect to the \emph{inclusion measure} of the Wagner hierarchy~\cite{Wagner75}. The inclusion measure is a semantic property --- it holds on all $\omega$-automata recognizing the same language, regardless if they are minimal, and regardless of their type (parity, Rabin, Streett, Muller, etc.).
Note that for all (deterministic) $\omega$-automata types the acceptance of a word $w$ is determined using some condition on $\inf(\rho_w)$, where $\rho_w$ is the run of the automaton on $w$. Thus, given an SCC $S$ we can classify it as \term{accepting} (resp. \term{rejecting}) if $\inf(\rho_w)=S$ implies $w$ is accepted (resp. rejected).
The \term{inclusion measure} considers the SCCs of the automaton graph (all SCCs, not just MSCCs), and counts the maximal number of SCCs in an inclusion chain of SCCs $S_1 \subset S_2 \subset ... \subset S_k$ where $S_i$ is accepting iff $S_{i+1}$ is rejecting for $1\leq i < k$, a so called \emph{alternating inclusion chain}.  Wagner has shown that for a given language $L$ this measure is the same for all $\omega$-automata recognizing $L$~\cite{Wagner75}, irrespective of the number of states or the automaton type~\cite{PerrinPinBook}, and is thus a semantic property. If the largest inclusion chain of an automaton for $L$ is $k$ and the innermost SCC is accepting (resp. rejecting), we say that $L\in\DMpk{k}$ (resp. $L\in\DMnk{k}$). If  $L\in\DMpk{k}$ and  $L\in\DMnk{k}$ we say  $L\in\DMpmk{k}$.

\autoref{fig:aut-for-examples-parity} shows the inclusion measure on parity automata (DPA) for three  examples.
We can see that  $\aut{P}_{1}\in\DMnk{2}$, $\aut{P}_{2}\in\DMpmk{1}$ and $\aut{P}_{3}\in\DMnk{4}$. 
Looking at $\aut{P}_{1}$ for $L_{\infty a}$ and at $\aut{P}_{2}$ for $L_{\ltlG a}$ and considering the words from \autoref{ex:refine-TN} and the preferences  
$\bullet\ \ b^i(a)^\omega \gtrbst_L b^{j}(a)^\omega$ \quad 
    $\bullet\ \ (b^ia)^\omega \gtrbst_L (b^ja)^\omega$ \quad 
    $\bullet\ \ (ba^j)^\omega \gtrbst_L (ba^i)^\omega$ \quad
    for $i<j$
from \autoref{req:TN} we can see that words that spend less time on the rejecting component (or reach it fewer times) are preferred over words that spend more time on the rejecting components (or reach it more times). The reader is invited to check that this is the case in all other discussed examples. 

Recall that we would like to  color the letters of the word to reflect whether they increase or decrease the robustness.
If we consider the ranks of the DPA  $\aut{P}_3$ for $L_{a\text{-seq}}$ we can see that letters on transitions that decrease the rank influence robustness. 
If they reach or stay in an even rank, they increase the robustness and so we would like to color them \col{green}. If they reach or stay in an odd rank they decrease robustness, and so we would like to color them \col{red}.
We can color letters on transitions that increase 
the rank \col{yellow} (for neutral effect on robustness).
The problem is that there may be different DPAs for a given language, even different minimal ones that yield different colors, as can be seen in 
\autoref{fig:dpa-infty-ab} (left three DPAs, discussed in more length in \autoref{sec:computing-the-score}). This conflicts with our wish to obtain a definition that is agnostic to the given representation. To overcome this, building on natural ranks for infinite words, we will first provide a semantic definition for the natural rank of an \emph{individual letter} of an infinite word.

\begin{figure}
\begin{center}
\scalebox{0.65}{
\begin{tikzpicture}[->,>=stealth',shorten >=1pt,auto,node distance=1.8cm,semithick,initial text=, initial left]

\node[state]    (e1)      {0};
\node[state]    (e0)   [below left of=e1]         {1};
\node[state]    (e2)  [below right of=e1]   {1};

\node[label] (qeL) [above left of=e1, node distance=1.25cm] {$\aut{P}_1:$};

\path (e0) edge [pred, loop below] 
           node   {$a$}
      (e0); 
\path (e0) edge [pgreen, bend left] 
           node   {$b$}
      (e1); 
\path (e1) edge [pyellow, bend left] 
           node   {$b$}
      (e2); 
\path (e2) edge [pgreen, bend left] 
           node   {$a$}
      (e1); 
\path (e1) edge [pyellow, bend left] 
           node   {$a$}
      (e0); 
\path (e2) edge [pred, loop below] 
           node   {$b$}
      (e2); 

\node[state]    (f1) [right of=e1, node distance=5.5cm]    {1};
\node[state]    (f0)   [below right of=f1]            {1};
\node[state]    (f2)  [below left of=f1]   {0};

\node[label] (qfL) [above left of=f1, node distance=1.25cm] {$\aut{P}_2:$};

\path (f0) edge [pred, loop below] 
           node   {$b$}
      (f0); 
\path (f0) edge [pred, bend right] 
           node   {$a$}
      (f1); 
\path (f1) edge [pred, loop above] 
           node   {$a$}
      (f1); 
\path (f1) edge [pgreen, bend right] 
           node [left]  {$b$}
      (f2); 
\path (f2) edge [pyellow, bend right] 
           node  [right] {$a$}
      (f1); 
\path (f2) edge [pyellow, bend right=45] 
           node   {$b$}
      (f0); 

\node[state]    (h1)  [right of=f1, node distance=5.5cm]   {1};
\node[state]    (h0)   [below right of=h1]            {1};
\node[state]    (h2)  [below left of=h1]   {0};

\node[label] (qhL) [above left of=h1, node distance=1.25cm] {$\aut{P}_3:$};

\path (h0) edge [pred, loop below] 
           node   {$b$}
      (h0); 
\path (h0) edge [pred, bend right] 
           node   {$a$}
      (h1); 
\path (h1) edge [pred, loop above] 
           node   {$a$}
      (h1); 
\path (h1) edge [pgreen, bend right] 
           node   {$b$}
      (h2); 
\path (h2) edge [pyellow, bend right=45] 
           node   {$a,b$}
      (h0); 

\node[state]    (j0)  [right of=h1, node distance=5.5cm]   {1};
\node[state]    (j1)   [below left of=j0, node distance=1.8cm]            {1};
\node[state]    (j2)  [below right of=j0, node distance=1.8cm]   {1};
\node[state]    (j3)  [below  of=j1, node distance=1.4cm]   {0};
\node[state]    (j4)  [below  of=j2, node distance=1.4cm]   {0};

\node[label] (qjL) [above left of=j0, node distance=1.25cm] {$\aut{P}^\truerobustness_{L_{\infty ab}}:$};

\path (j0) edge  [pred]
           node [above]  {$a$}
      (j1); 
\path (j0) edge  [pred]
           node  [above] {$b$}
      (j2); 
\path (j1) edge [pred, loop left] 
           node    {$a$}
      (j1); 
\path (j2) edge [pred, loop right] 
           node    {$b$}
      (j2); 
\path (j1) edge  [pgreen]
           node   [left] {$b$}
     (j3);
\path (j2) edge  [pgreen]
           node   [right] {$a$}
      (j4);
\path (j4) edge  [pgreen, bend right]
           node   [above] {$b$}
      (j3);
\path (j3) edge  [pgreen]
           node  [below]  {$a$}
      (j4);
\path (j4) edge  [pyellow, bend right]
           node  [right] {$a$}
      (j1);
\path (j3) edge  [pyellow, bend left]
           node  [left]  {$b$}
      (j2);

\end{tikzpicture}}
\end{center}
\vspace{-9mm}
\caption{Parity automata $\aut{P}_1,\aut{P}_2,\aut{P}_3,\aut{P}_{L_{\infty ab}}^\truerobustness$ for $L_{\infty ab}$.}\label{fig:dpa-infty-ab} 
\end{figure}
\vspace{-4mm}

\subsubsection*{Natural ranks of $\omega$-words and their infixes (necessary background)} 
\vspace{-1mm}
There is a tight correlation between the minimal number of ranks used by a DPA, and the inclusion measure of the language~\cite{NiwinskiW98,CartonM99,PerrinPinBook}. 
Based on this correlation Ehlers and Schewe have managed to define the semantic notion of \term{natural ranks of infinite words}~\cite{EhlersS22}. 
Given an $\omega$-regular language $L\subseteq \Sigma^\omega$ 
and a word $w\in\Sigma^\omega$, its \emph{natural rank} is a natural number in a range $\{0,1,...,k\}$ that reflects the minimal (state-)rank it can receive in a DPA for $L$ whose (state-)ranks are as low as possible.
We use $\mathterm{\wordrank_L(w)}$ for the natural rank of $w$ wrt $L$. 

We briefly outline the core intuitions behind natural ranks of infinite words and refer unfamiliar readers to~\cite{EhlersS22} for a more comprehensive treatment. The definition of a natural rank of a word, follows the intuition of Niwinski and Walukiewicz in their construction for minimizing 
\begin{wrapfigure}[4]{r}{0.14\textwidth}
\vspace{-8mm}
\begin{center}
\scalebox{0.65}{
\begin{tikzpicture}[->,>=stealth',shorten >=1pt,auto,node distance=2.0cm,semithick,initial text=, initial above]

\node[state, minimum size=2ex]   (q0)  {$q$};

\path (q0) edge [loop left] 
           node   {$b | 3$}
      (q0); 
\path (q0) edge [loop above] 
           node   {$ab|2$}
      (q0); 
\path (q0) edge [loop right] 
           node   {$aab|1$}
      (q0); 
\path (q0) edge [loop below] 
           node   {$aaab|0$}
      (q0);       
\end{tikzpicture}}
\end{center}
\end{wrapfigure}
the number of ranks in a
parity automaton~\cite{NiwinskiW98}. 
The observation is that if a DPA requires $n$ ranks then one can
reveal in its transition graph a structure resembling a flower, in the sense that it has a center state $q$ from which 
there are at least $n$ loops (forming the petals of
the flower) such that the minimal ranks visited along the different petals form
a sequence $c_1 < c_2 < ... < c_n$ of ranks with alternating parity. For instance,
in $\aut{P}_3$ of \autoref{fig:aut-for-examples-parity} there is a flower with center state $3$ and petals $b$, $ab$, $aab$, $aaab$ with ranks $3$, $2$, $1$, $0$, respectively. This flower structure is useful for understanding the natural ranks as we explain next.

The following notations are needed for the formal definition of a natural rank. A word $z\in\Sigma^*$ is said to be a \emph{suffix invariant} of $u\in\Sigma^*$ with respect to $L$ if $uz\sim_L u$. That is, no suffix distinguishes between $u$ and the word obtained by concatenating $z$ to $u$. 
A word $w'$ is \term{obtained from $w=\sigma_1\sigma_2\sigma_3...$ by an infinite series of injections} of suffix invariant words $z_1,z_2,z_3,...$ at positions $i_1,i_2,i_3,...$ if $w'=\sigma_1\sigma_2... \sigma_{i_1}z_1 \sigma_{i_1+1}\sigma_{i_1+2}... \sigma_{i_2}z_2\sigma_{i_2+1}...$, where $z_j$ is suffix invariant wrt $\sigma_1...\sigma_{i_j}$.
We use $\mathterm{\inj_I(w)}$ for the set of words $w'$ obtained from $w$ by such an infinite series of injections at positions in $I=i_1,i_2,...$.  The injections are the semantical way to capture the said flower structure. When such a flower structure exists with center state $q$, reached by word $u$, and $v_1,v_2,...,v_k$ are petals, then $uv_i\sim_L u$ for all $i\in [1..k]$ and if $w=u(v_i)^\omega$ then $w'=u(v_iv_j)^\omega$ is in $\inj_I(w)$ for all $j\in[1..k]$ and the respective $I$.

The formal definition of a natural rank is inductive, starting with $0$ and $1$ at the base (where $1$ can also be obtained in the inductive step). Loosely speaking, a word $w$ has natural rank $0$ if it is in $L$, as are all words $w'\in\inj_I(w)$ for some $I\subseteq \mathbb{N}$. Similarly, a word $w$ has natural rank $1$ if it is not in $L$ and there exists $I\subseteq \mathbb{N}$ such that 
all $w'\in\inj_I(w)$ either have rank $0$ or are not in $L$.
Inductively, a word has natural rank $k$ if  there exists $I\subseteq \mathbb{N}$ such that $w$ is in $L$ (resp. not in $L$) and $k$ is even (resp. odd) and all words $w'\in\inj_I(w)$ have a smaller natural rank or are also in (resp. not in) $L$, and there exists at least one word $w'\in\inj_I(w)$ that has natural rank $k{-}1$. 
Below we give a slightly modified definition, with two additional base cases $-2,-1$. Their importance will be made clear in the sequel.
\vspace{-2mm}
\begin{definition}[Natural rank of an infinite word~\cite{EhlersS22}, slightly altered]\label{def:inf-nat-rank-minuses}We define the \term{natural rank} of an infinite word $w$ wrt $L$, denoted $\mathterm{\wordrank_L(w)}$ (or simply $\wordrank(w)$ when $L$ is clear from the context)
as follows. 
\vspace{-2mm}
{
{
$$\mathterm{\wordrank_L(w)}=\left\{\begin{array}{r@{\quad}l}
-2 & \text{if } \exists u\prec w. \forall z\in \Sigma^\omega.\ uz \in L
\\
-1 & \text{else if } \exists u\prec w. \forall z\in \Sigma^\omega.\ uz\notin L
\\
0 & \text{else if } w\in L \text{ and } \exists I\subseteq \mathbb{N} \text{. } \forall w'\in\inj_I(w).\ w'\in L
\\
1 & \text{else if } w\notin L \text{ and } \exists I\subseteq \mathbb{N} \text{. } \forall w'\in\inj_I(w).\ w'\notin L 
\\
i & \text{else if } i = \min\left\{j~\left|~\begin{array}{l} 
j \text{ is even iff } w\in L \text{ and } \exists I\subseteq \mathbb{N}.  \\
\forall w'{\in}\inj_I(w).\ (\wordrank_L(w')< j \text{ or } \\ \phantom{-----....} w'\in L \iff w \in L) 
\end{array}\right.\right\}
\\
\end{array}
\right.$$
}}
\end{definition}

\vspace{1mm}
\begin{example} 
    On the examples whose DPA are given in \autoref{fig:aut-for-examples-parity}, the following holds. 
\end{example}

\begin{center}
{\small
\begin{minipage}{0.25\textwidth}
\setlength{\tabcolsep}{1.6pt}
\hspace{5pt}\begin{tabular}{l@{\quad}r@{\ }|@{\ }c}
& word $w$ & $\wordrank_{L_{\infty a}}(w)$\\
\cline{2-3}
 & $(b)^\omega$   &  $1$\\ 
  &  $aa(b)^\omega$   &  $1$\\ 
 & $(a)^\omega$   &  $0$\\ 
 &   $(ab)^\omega$   &  $0$\\ 
\end{tabular}
\vspace{2mm}
\end{minipage}
\begin{minipage}{0.25\textwidth}
\setlength{\tabcolsep}{1.6pt}
\begin{tabular}{l@{\quad}r@{\ }|@{\ }c}
& word $w$ & $\wordrank_{L_{\ltlG a}}(w)$\\
\cline{2-3}
 & $(a)^\omega$   &  $0$\\ 
 &   $b(a)^\omega$   &  $-1$\\ 
 & $(ab)^\omega$   &  $-1$\\ 
 &   $(b)^\omega$   &  $-1$\\ 
\end{tabular}
\vspace{2mm}
\end{minipage}
\begin{minipage}{0.44\textwidth}
\setlength{\tabcolsep}{1.6pt}
\begin{tabular}{l@{\quad}r@{\ }|@{\ }c}
& word $w$ & $\wordrank_{L_{a\text{-seq}}}(w)$\\
\cline{2-3}
  &      $(b)^\omega$   &  $3$\\ 
   & $(ab)^\omega$   &  $2$\\ 
 & $(aab)^\omega$   &  $1$\\ 
 & $(a)^\omega,(aaab)^\omega$   &  $0$\\ 

\end{tabular}
\vspace{2mm}
\end{minipage} 
}
\end{center}
\vspace{-2mm}

See \autoref{remark-exists-I} on the usage of the \emph{existence} of a sequence $I$ of injection points in
\autoref{def:inf-nat-rank-minuses}.

Note that the natural ranks alone do not suffice to obtain the desired robustness preference order. E.g. wrt $L_{\infty a}$ all words $(a)^\omega$, $bbb(a)^\omega$ and $(ab)^\omega$ have the same natural-rank $0$ while we want to distinguish them and have $(a)^\omega \gtrbst_{L_{\infty a}} bbb(a)^\omega \gtrbst_{L_{\infty a}} (ab)^\omega$. We need a way to capture that $(a)^\omega$ stays forever within the accepting part of its SCC, whereas $bbb(a)^\omega$ spends some time at the subsumed rejecting SCC and $(ab)^\omega$ crosses to a subsumed rejecting component over and over.

To this aim we would like to associate a rank to the letters in a word, so that we can  distinguish letters that remain in the same accepting SCC and letters that cross from an accepting SCC to a subsumed rejecting one. In a sense, we  ask for any $k\in\mathbb{N}$ whether taking the $k$-th letter is desirable, where we view the $k$-th letter as part of a loop that the entire $\omega$-word traverses infinitely often. We 
build on
the definition of a natural 
rank of a period starting with $v$ wrt a given spoke $u$ as the maximal rank that such a period can get if it closes a loop on the spoke $u$~\cite{BohnL24}, with some adaptation to non-negative and immaterial ranks.  

\begin{definition}[Natural rank of an infix~\cite{BohnL24}, slightly altered]\label{def:infix-rank}
Let $u\in\Sigma^*$, $v\in\Sigma^*$. 
Let $L_{u.v}=\{u(vz)^\omega~|~vz\in\Sigma^+, uvz \sim_L u \}$.
The \term{rank of a period starting with $v$ following a spoke $u$}, denoted 
$\mathterm{\infixrank_L(u,v)}$, is
\vspace{-2mm}
\[
\mathterm{\infixrank_L(u,v)}=\left\{
    \begin{array}{l@{\ \ \ }l}
       -2  & \textrm{if }  \forall w\in\Sigma^\omega.\ uvw \in L  \\
       -1  & \textrm{if } \forall w\in\Sigma^\omega.\ uvw \notin L  \\
       \max\ \{ \wordrank_L(w)~|~w\in L_{u.v} \}  &
    \text{if } L_{u.v}\neq\emptyset \\
        \infty & \textrm{otherwise} \textrm{ (if } L_{u.v}=\emptyset \textrm{)}
    \end{array}
\right.
\]
\vspace{-3mm}
\end{definition}

\begin{example}
Wrt $L_{a\text{-seq}}$ we have that $\infixrank(\varepsilon,b)=3$ as witnessed by $z=\epsilon$. $\infixrank(\varepsilon,a)=2$ as witnessed by  $z=b$. $\infixrank(\varepsilon,aa)=1$ as witnessed by  $z=b$. And, $\infixrank(\varepsilon,aaa)=0$ as witnessed by  $z=\epsilon$.
Note that $\infixrank$ is more informative than $\wordrank$ since $(a)^\omega=(aa)^\omega=(aaa)^\omega=a^\omega$, hence $\wordrank$ cannot distinguish these. Yet the different sequences of $a$'s have different $\infixrank$. 
\end{example}

\begin{remark}\label{rem:L-mod-2}
    In general, the rank of an infix depends on the equivalence class of the spoke $u$. In $L_{a\text{-seq}}$ there is a single equivalence class so this is not apparent. 
An example for a language with two equivalence classes is
    $L_{\text{mod}2}=\infty a~\vee~(|w|_{a}$ is even $\wedge~\neg \infty c)~\vee~(|w|_{a}$ is odd $\wedge~\neg \infty b)$ over $\Sigma=\{a,b,c\}$ where $|w|_a$ is the number of occurrences of the letter $a$ in $w$.
    \autoref{ex:infix-rank-mod-2} provides an analysis of  $\infixrank$ for $L_{\text{mod}2}$.
\end{remark}

\begin{restatable}[]{remark}{rmrkinfixnoincrese}\label{rmrk:infix-no-increse}
Note that reading a longer prefix of the period, its rank can only decrease, that is, $\infixrank_L(u,vy)\leq\infixrank_L(u,v)$. Moreover, if $u\sim_L xy$ then also $\infixrank_L(x,yv)\leq \infixrank_L(u,v)$~\cite{BohnL24}. 
\end{restatable}

\section{Defining the robustness value}\label{sec:define-rbst}

Recall that we want to associate with each letter of the word a color that reflects whether it increases or decreases the robustness value. Consider the language $L_{\infty a}$ and the word $w=ab(bbaabba)^\omega$. We expect to get that there are three greens in the period. 
Trying to derive the color of the $k$-th letter from the rank of an infix including the $k$-th letter, we encounter the following problem. 
If the $k$-th letter is $b$ it will be ranked $1$ if the part of the period up to $k$ has no $a$'s and $0$ otherwise.  Accordingly, for $w=ab(bbaabba)^\omega$ the ranks of the period $bbaabba$
will be $1100000$ which yield $5$ green letters while there are only three $a$'s in the period. The problem occurs since ranks along a period are monotonically non-increasing (as per \autoref{rmrk:infix-no-increse}), and so once an $a$ is observed in the period all subsequent letters are ranked $0$ and their role in robustness is lost.

Hence, in order to get some meaningful information on the current letter we have to forget the past from time to time, and 
assume the period has recently started even if it started a long while ago.
That is, we want to "reset" the calculation of ranks from time to time. If we reset in very long intervals, we still get the effect that the sequence of ranks from the last reset point is not very informative regarding late indices. If we reset very often we may get a result that is higher than the true rank of the index we are currently looking at. 
The challenge is thus to find the best places to make a reset, so that the observed rank provides the correct information regarding the robustness of the desired index. 

Put otherwise, for each index $k$ we'd like to determine what is the \term{dominant suffix} of $w[..k]$,  the suffix wrt which we should judge the contribution of the current index to the robustness. If it is $w[j..k]$ then $j$ is our reset point. For instance, in the language $\infty abcd$ the dominant suffix of index $k$ would be a non-empty prefix of $abcd$ if such a prefix is a suffix of $w[..k]$. Otherwise, it would be its last letter, $w[k]$. 
For $L_{a\text{-seq}}$, if $w[..k]$ is in $\Sigma^*b$ it suffices to know that it ends with a $b$, and so the dominant suffix is $b$. If $w[..k]$ ends with $a^i$, then if $i\geq 3$ it suffices to know that there is a suffix of three $a$'s, but we definitely do not want to shorten a suffix of three $a$'s to just one $a$. So the dominant suffix should be $aaa$ if $i\geq 3$; $aa$ if $i{=}2$; and $a$ if $i{=}1$. We turn to develop the notions that will serve us in defining this formally.

\subsection{Defining natural ranks of letters}
Considering some infix of an infinite word, whose rank is, say $d$, we can say that its last index is \term{influential} if there exists a suffix of the given infix for which the last index makes a difference, namely it decreases the rank
to $d$.
Formally, this is defined as follows where $\mathterm{\infixrank_L(w,j,k)}$ is used as a shorthand for $\infixrank_L(w[..j),w[j..k])$

\begin{definition}[Influential index]\label{def:influential}
Let $w\in\Sigma^*\cup\Sigma^\omega$ and $L\subseteq \Sigma^\omega$.  
We say that index $k$ of $w$ is \term{influential} wrt index $i<k$ (and $L$) if there  exists  $j\in[i..k]$ such that
$$\begin{array}{l}
\infixrank_L(w,j,k) < \infixrank_L(w,j,k{-}1)\qquad \text{ and } \qquad
\infixrank_L(w,j,k) = \infixrank_L(w,i,k).
\end{array}$$ 
\end{definition}

For instance in the language $\infty abc$, the $k$-th letter of $w$ is influential iff $w[k{-}2..k]=abc$. 
In  $L_{a\text{-seq}}$ the $k$-th letter is influential iff it is $a$. See calculations in \autoref{ex:influential}.

We turn to define the dominant suffix and the dominant index, the index from which the dominant suffix begins. If the $k$-th index is influential, namely it decreases the 
{rank} with regard to some suffix, it suffices to take the shortest suffix that has the same 
{rank} as the entire infix $w[i..k]$. This is the case for instance in $L_{a\text{-seq}}$ where we'd like the dominant suffix of $bba$ to be $a$ and of $bbaaaaa$ to be $aaa$.

If the $k$-th index is not influential, then the {infix rank} of all suffixes $w[j..k]$ is the same as the  {infix rank} of the respective suffix $w[j..k)$.
This is the case e.g. in $L=\infty abcd$ and $w[i..k]=abc$ where the rank of all the suffixes $abc$, $bc$, $c$ is $1$.
When the $k$-th letter is non-influential, the  {infix rank} of $w[k]$ can be higher than that of the entire infix $w[i..k]$ which may include influential letters, and so we should reset to a suffix $w[j..k]$ with the 
{rank} of $w[k]$. E.g. in $L_{a\text{-seq}}$ and $w[i..k]=baab$ the rank of $baab$ is $1$ whereas the rank of $w[k]=b$ is $3$ which is higher.
Here we reset to the 
{shortest} such suffix that has all necessary information about the letters read so far that might be meaningful in the future, in order not to lose information.  E.g. in $L_{a\text{-seq}}$ and $baab$ it would be $b$; in $\infty abcd$ for $abc$ it would be $abc$ and for $cba$, it is $a$.
The formal definition makes use of the intermediate notion of a \emph{predominant-suffix}.

\begin{definition}[(Pre)dominant suffix/index]\label{def:dom-ind}
Let $w\in\Sigma^*\cup\Sigma^\omega$. 
Let $i \leq k \leq|w|$.
\begin{enumerate}
    \item If $k$ is influential wrt $i$ then the \emph{shortest} suffix of $w[i..k]$ for which  
    $\infixrank_L(w,j,k)=\infixrank_L(w,i,k)$ is its \term{predominant-suffix} as well as its \term{dominant-suffix}.
    \item Otherwise, the \emph{longest} suffix $w[l..k]$ of $w[i..k]$ for which $\infixrank_L(w,l,k)=\infixrank_L(w,k,k)$ is its \term{predominant-suffix}. The \emph{shortest} suffix of $w[l..k]$ for which $\infixrank_L(w,j,k)=\infixrank_L(w,l,k)$ and $\infixrank_L(w[..j),w[j..k]\cdot v )=\infixrank_L(w[..l),w[l..k]\cdot v)$ for all $v\in\Sigma^*$ is its \term{dominant-suffix}.
\end{enumerate}
The first index of the dominant (resp. predominant) suffix is termed the dominant (resp. predominant) \term{index} and is denoted $\mathterm{\domind_L(w,i,k)}$ and $\mathterm{\predomind_L(w,i,k)}$.
When $w\in\Sigma^*$ we use $\mathterm{\domind_L(w,i)}$ as a shorthand for $\domind_L(w,i,|w|)$ and similarly for predominant.
\end{definition}

From this definition we can associate with each index $k$ an index $j_k \leq k$ that is the \emph{reset point} for $k$. Intuitively, the natural rank of letter $k$ (as we shortly define) 
can be computed from the infix $w[j_k..k]$ rather than the entire prefix $w[..k]$.

\begin{definition}[Reset points]\label{def:ref-indices}
Let $w\in\Sigma^*\cup\Sigma^\omega$ and $k$ a natural number satisfying $k\leq|w|$.
We define the \term{reset point} of $k$ wrt $w$ and $L$, denoted $\mathterm{\reset_w(k)}$ inductively as follows.
For the base case $\reset_w(0)=0$. For $k>0$, we let $\reset_w(k)=\domind_L(w,\reset_w(k{-}1),k)$.
\end{definition}

The following claim asserts that computing the dominant index of $k$ wrt 
any $i < \reset_w(k-1)$ is the same as wrt $\reset_w(k-1)$.

\begin{restatable}[]{myclaim}{lemds}\label{lem:ds}\label{cor:domind}
    $\domind_L(w,i,k)=\domind_L(w,\reset_w(k{-}1),k)$
    for all $k\in\mathbb{N}$ and  $i\leq\reset_w(k{-}1)$.
    Thus in particular, $\domind_L(w,0,k)=\domind_L(w,\reset_w(k{-}1),k)=\reset_w(k)$.
\end{restatable}

If $v{\in}\Sigma^*$, $\mathterm{\domsuf_L(v)}$ and $\mathterm{\predomsuf_L(v)}$ stand
    for
    $v[r..]$ and $v[r'..]$ resp.
    where $r{=}\reset_v(|v|)$ and $r'{=}\predomind_L(v,0)$.    

Using \autoref{lem:ds} on the definition of dominant suffix we obtain: 

\begin{corollary}\label{cor:domsuf}
     $\domsuf_L(\domsuf_L(x){\cdot}\sigma)=\domsuf_L(x\sigma)$ for every $x\in\Sigma^*$ and $\sigma\in\Sigma$.
\end{corollary}

We are now ready to define the natural rank of the $k$-th letter of an infinite word $w$. It  is the rank of the dominant suffix of $w[..k]$.
See \autoref{ex:nat-rank-calcs} for some examples.

\begin{definition}[Natural rank of a letter]\label{def:letter-rank}
Let $w\in\Sigma^*\cup\Sigma^\omega$ and $k\leq|w|$.
The \term{natural rank of the $k$-th letter} of $w$, denoted $\mathterm{\letterrank_L(w,k)}$, is $\infixrank_L(w,\reset_w(k),k)$. For a finite word $v\in\Sigma^+$ we use   $\letterrank_L(v)$ as an abbreviation for $\letterrank_L(v,|v|)$.
\end{definition}

The following theorem states a correlation between the minimal natural rank of the letters that occur infinitely often and the natural rank of the entire infinite word.

\begin{restatable}[Natural rank word/letters correlation]{theorem}{thmnatrankoflettersandwords}\label{thm:natrank-of-letters-and-words}
    Let $L \subseteq \Sigma^\omega$ and $w\in\Sigma^\omega$.
    For $i\in\mathbb{N}$ let  $n_i=\letterrank_L(w,i)$.
    Then $\wordrank_L(w)=d$ if and only if $\min \inf (n_1,n_2,...) = d$.
\end{restatable}

\subsection{Defining colors and scores of letters and infixes}
Now that each letter has its own rank, we can color the letters of the word following our intuitions from the Wagner hierarchy. 
Let $w=\sigma_1\sigma_2...$ be an infinite word.
We would like to give the $i$-th letter of the word a \emph{color} reflecting its transition between accepting/rejecting SCCs, and whether this is done in the direction of the inclusion chain or opposite to it. As discussed in the third paragraph of \autoref{sec:toward} considering ranks this suggest the following scheme.
Let $d=\letterrank_L(w,i)$ and 
$d'=\letterrank_L(w,i{+}1)$. We can \emph{try} saying that the $i{+}1$-th letter has color
\begin{inparaitem}[\textcolor{lipicsBulletGray}{$\bullet$}]
    \item 
    \term{green} if $d'$ is even and $d \geq d'$,
    \item \term{yellow} if $d < d'$,
    \item \term{red} if $d'$ is odd and $d \geq d'$.
\end{inparaitem}

But this coloring scheme doesn't yet match the intuition from all of our examples. Consider the language $L_{\ltlG a}$ and its automaton $\aut{P}_{2}$ (\autoref{fig:aut-for-examples-parity}). 
 We would like to say that $ba^\omega \gtrbst b^{5}a^\omega$  but there is no difference in the colors of their letters --- 
 in both words, after taking the first letter, we reach a rejecting sink, and accordingly the rest of the letters are red all along, whereas we expected the coloring to be \rdl{$b$}$($\grl{$a$}$)^\omega$ and \rdl{$bbbbb$}$($\grl{$a$}$)^\omega$. 
 
This is since reading the first $b$, the property is already violated.  
Namely the $b$ prefix cannot be recovered, even if the extension is $a^\omega$, namely there was just one fault at the beginning and then everything was okay.
For this reason we added the natural ranks $-1$ and $-2$ where $-1$ reflects that no extension of the read prefix is in the language, and $-2$ reflects that all extensions of the read prefix are in the language.
Using these ranks we define the color of a letter, building on the following notion of the \emph{forgetful-version} of a word.

\begin{definition}[Forgetful version of a word]
    For a finite or infinite word
    $w=\sigma_1\sigma_2...$ its \term{forgetful version} wrt $L$, denoted $\mathterm{\fg_L(w)}$ (or simply $\fg(w)$ when $L$ is clear from the context) is defined inductively as follows:
$\fg_L(w)=\sigma'_1\sigma'_2\sigma'_3...$  where
\begin{center}

$\begin{array}{l@{\ =\ }l@{\quad}l@{\ =\ }l}
\sigma'_1 & \left\{ 
\begin{array}{l@{\quad}l}
    \sigma_1 & \text{if } \letterrank_L(\sigma_1)\geq 0  \\
     \varepsilon & \text{otherwise }
\end{array}
\right. 
& 
\sigma'_{i+1} & \left\{ 
\begin{array}{l@{\quad}l}
    \sigma_{i+1} & \text{if } \letterrank_L(\sigma'_1\cdots\sigma'_i\cdot \sigma_{i+1})\geq 0  \\
     \varepsilon & \text{otherwise }
\end{array}
\right.
\end{array}$
\end{center}
\end{definition}
For instance, the forgetful version of both $ba^\omega$ and $b^{5}a^\omega$ wrt $L_{\ltlG a}$ is $a^\omega$ since $\wordrank(bw)=-1$  for every $w\in\Sigma^\omega$. While this seems to forget too much, the following revised definition of colors of letters will not forget these $b$'s, and will yield
\bll{$b$}$($\grl{$a$}$)^\omega$ and \bll{$bbbbb$}$($\grl{$a$}$)^\omega$ similar to our expectation (just using \col{black} rather than \col{red}).

\begin{definition}[Color of a letter]\label{def:letter-color}
Let $w=\sigma_1\sigma_2\cdots\in\Sigma^
\omega$. 
Let ${d}=\letterrank(\fg(w[..i))\cdot \sigma_i)$ and  ${d'}=\letterrank(\fg(w[..i])\cdot\sigma_{i+1})$.
The \term{$i{+}1$-th letter has color} $\theta(d,d')$ where $\theta(d,d')$ is  \vspace{-2mm}
\[
\begin{array}{@{\qquad\qquad}l@{\ }l@{\qquad}l@{\  }l@{\qquad}l@{\ }l}
     \text{\term{white}} & \text{if } d'=-2 &
     
     \text{\term{green}} &\text{if } d' \text{ is even and } d \geq d' &          
     
     \text{\term{yellow}} & \text{if } d < d' \\
     
     \text{\term{black}} & \text{if } d'=-1 & 
     
     \text{\term{red}}  & \text{if } d'  \text{ is odd and } d \geq d'
\end{array}
\]
\vspace{-2mm}
We use 
$\mathterm{\lcolor_w(i)}$  
for the color of the $i$-th letter in $w$ (wrt $L$ which is implicit). 
\end{definition}

Given an infix of an infinite word  we can count how many letters from each color it has.

\begin{definition}[Infix colors]\label{def:infix-scr}
Let $w\in\Sigma^\omega$ and consider its infix $w[j..k]$.
    Let $\mathterm{c_i}=\lcolor_w(i)$ be the color of the $i$-th letter for $i\in[j..k]$. Let
    \vspace{-2mm}
    \[\begin{array}{l@{\ =\ }l@{\ \ \ }l@{\ =\ }l@{\ \ \ }l@{\ =\ }l}
        \mathterm{\white} & |\{i\in[j..k]\colon c_{i} = \col{white}\}|
        &
        \mathterm{\green} & |\{i\in[j..k] \colon c_{i} = \col{green}\}|
         & \mathterm{\yellow} & |\{i\in[j..k] \colon c_{i} = \col{yellow}\}|
        \\
         \mathterm{\black} & |\{i\in[j..k] \colon c_{i} = \col{black}\}| 
         &
        \mathterm{\red} & |\{i\in[j..k] \colon c_{i} = \col{red}\}|\\
    \end{array}\]
We use 
$\mathterm{\icolors_w(j,k)}$, 
for the tuple $(\white,\green,\yellow,\red,\black)$ providing the number of letters of each color in the infix $w[i..j]$ (wrt $L$ which is implicit).
\end{definition}

Using the colors of an infix we can define the \emph{score of the infix}. To this aim we note that \col{green} and \col{red} letters have the opposite influence on the acceptance of the word. In a sense, they cancel each other, and so the difference between the number of \col{green}s and the number of \col{red}s is what matters.
The same is true for  \col{white} and \col{black} letters. 
Accordingly we define:\begin{definition}[Score of an infix]\label{def:infix-scr}
    Let $w{\in}\Sigma^\omega$ and $j,k{\in}\mathbb{N}$.
    Let $\icolors_w(j,k)=(\white,\green,\yellow,\red,\black)$.
    The \term{score of the infix} $w[j..k]$ wrt $L$, denoted $\mathterm{\score_w(j,k)}$, is the tuple $(\white\black,\green\red)$ where $\mathterm{\white\black}=\white{-}\black$, $\mathterm{\green\red}=\green{-}\red$. Its \term{averaged score}, denoted $\mathterm{\avgscore_w(j,k)}$ is $(\frac{\white\black}{l},\frac{\green\red}{l})$ where $l=|w[j..k]|$.
\end{definition}

\subsection{The robustness value definition and properties}

As we shortly define, the robustness value is a tuple of numbers from a domain $\mathbb{T}$, where elements are compared component-wise and the leftmost component is the most significant. Let $w\in\Sigma^\omega$ be a lasso word and assume $(u,v)$ is a good decomposition of it.\footnote{We will shortly define what is a good decomposition.}
The abstract form of the robustness value of $w$ is  $(a_w, s_v, s_u)$ where $a_w$, aka the \term{acceptance bit} is $1$ if $w\in L$ and $-1$ otherwise; $s_v$ is a tuple reflecting the \term{period-value} and $s_u$ is a tuple reflecting the \term{spoke-value}. Obviously, the value should first take into account whether the word is accepted or not; and among the period and spoke, the period should be more significant.  

Note that the second and third components are the period's and spoke's \emph{values}, not their \emph{scores}. Moreover, as we  explain next, the value of the spoke is computed in a different manner from the value of a period. For the value of the period, we would like $v$ and $vv$ to receive the same value, since they induce the same infinite word. Accordingly, the value of the period is  its averaged score.  

For the value of the spoke, things are a little bit more complicated. Here we would like to take into account the length of the spoke. But whether a shorter or a longer spoke is preferred depends on the relation between the averaged scores of the period and spoke. 
Consider the language $L_{\infty a}$, and the words $w_1=bb(ab)^\omega$, $w_2=bbbbb(ab)^\omega$. Clearly, $w_1$ with the shorter spoke should be determined more robust than $w_2$ with the longer spoke. But considering the same language and the words
$w_3=aa(ab)^\omega$ and $w_4=aaaaa(ab)^\omega$, obviously $w_4$ with the longer spoke should be determined more robust than $w_3$ with the shorter spoke. Accordingly, the value of the spoke considers the difference between the averaged scores of the spoke and period, multiplied by the length of the spoke. Put otherwise, it asks for each letter of the spoke how much better or worse it is compared to the averaged score of the period, and sums these values.

Note that a given lasso word has many decompositions into a spoke and a period, and different decompositions may yield different results. E.g. for $L_{a\text{-seq}}$ and $w=a^\omega$ we would like a decomposition where the period has $aaa$. If it has a shorter sequence of $a$'s the rank of the period would be wrong. For the spoke we would like the shortest spoke reaching a correct period. We thus define:

\begin{definition}[The robustness decomposition]\label{def:rbst-dcmp}
Let $w\in\Sigma^\omega$ be a lasso word.
Its \term{robustness decomposition} wrt $L$, denoted $\mathterm{\rbstdecomp_L(w)}$ is the shortest decomposition $(u,v)$ satisfying $w=uv^\omega$, $\infixrank_L(u,v)=\wordrank_L(w)$.\footnote{I.e., for every  $x,y$ satisfying $w=xy^\omega$ and $\infixrank_L(x,y)=\wordrank_L(w)$ we have $|uv|<|xy|$.}
\end{definition}
We are finally ready to define the desired robustness value.

\begin{definition}[The robustness value]\label{def:rbst-scr} 
Let $\rbstdecomp_L(w)=(u,v)$ where $|u|=k$ and $|v|=l$. 
Let $\tau_u=\avgscore_w(1,k)$, 
$\tau_v=\avgscore_w(k{+}1,k{+}l)$, and $a_w$ the acceptance bit.
The \term{robustness value of $w$ wrt $L$}, denoted $\mathterm{\valrbst_L(w)}$ is the tuple $\mathterm{(}a_w\mathterm{,}\ \tau_v\mathterm{,}\ k(\tau_u{-}\tau_v)\mathterm{)}$. 
\end{definition}

\autoref{ex:rbstval} provides examples of words, languages and the respective robustness value. The reader can verify that it matches our intuition on all provided examples.

The robustness values are elements in $\{-1,1\}\times[-1,1]^2\times \mathbb{Q}^2$.
Letting $\mathbb{T}=\{-1,0,1\}\times[-1,1]^2\times \mathbb{Q}^2$
and $\zero=(0,(0,0),(0,0))$, it is immediate that the value function $\valrbst_L$  refines $L$,
since its most significant component is the acceptance bit,
which is positive iff $w\in L$.

\begin{restatable}[Satisfaction bar]{myclaim}{clmsatisfactionbar}\label{clm:satisfaction-bar}
    $\valrbst_L(w)\geq \zero$ iff $w\in L$.
\end{restatable}

It follows from \autoref{def:infix-scr}
that the score of an infix $w[j,k]$ wrt a language $L$ is the negative of the score of $w[j,k]$ wrt $\bar{L}$, the complement of $L$. 
Obviously, this holds also for the averaged-score. 
Since this is also true for the acceptance bit, 
\autoref{def:rbst-scr} guarantees:

\begin{restatable}[Value wrt complement]{myclaim}{propscoreforLcomplement}\label{prop:score-for-L-complement}
    $\valrbst_L(w)=-\valrbst_{\bar{L}}(w)$ 
    for every $w$ and $L$.
\end{restatable}

Using the robustness value we can define the robustness preference relation.
\begin{definition}[The robustness preference relation]\label{def:rtst-prf-relation}
The \term{robustness preference relation} is defined as follows.
    We say that $\mathterm{w_1 \geqrbst_L w_2}$ 
    if $\valrbst_L(w_1)\geq \valrbst_L(w_2)$ where $\geq$ works component-wise and the left-most coordinate is the most significant.
\end{definition}

The duality theorem as required by \autoref{req:duality} clearly follows from Claim~\ref{prop:score-for-L-complement}.

\begin{restatable}[Duality]{myclaim}{clmduality}\label{clm:duality}
    Let $L\subseteq \Sigma^\omega$ be an $\omega$-regular property, and $\bar{L}$ its complement. Let $w_1,w_2\in\Sigma^\omega$. Then
     $w_1 \,\geqrbst_L\, w_2 \iff w_1 \,\leqrbst_{\bar{L}}\, w_2.$
\end{restatable}

The robustness value function and preference relation thus meet our intuitions and adhere the required mathematical criteria.

\begin{restatable}[The robustness preference relation]{corollary}{thmrbstprefgreat}\label{thm:rbst-pref-great}
The robustness preference relation $\geqrbst_L$ induced from $L$ satisfies all our requirements.
\end{restatable}

\section{Computing the robustness value}\label{sec:computing-the-score}
We turn to ask how to calculate the robustness value of certain words with respect to a certain language. Ideally, we would like to have a DPA whose transitions can be colored in the same manner as we colored letters of the word so that there is a one-to-one match between the colors of the edges traversed during the run and the colors of the letters of the word.
We show that such a DPA exists and the correspondence is tighter than just the colors: the ranks of states traversed while reading the word match exactly the ranks of the letters.

As mentioned in \autoref{sec:toward},
taking an arbitrary DPA for the language, even a minimal one with minimal ranks, 
may not produce the desired correlation. Indeed, \autoref{fig:dpa-infty-ab} shows three minimal DPAs with minimal ranks for the language $L_{\infty ab}$: $\aut{P}_1$, $\aut{P}_2$ and $\aut{P}_3$, but they often disagree on the sequence of ranks traversed when reading the same word.  For instance, $(ab)^\omega$ gets ranks $(10)^\omega$ 
and colors $($\yll{$a$}\grl{$b$}$)^\omega$ on $P_1$ and $P_2$ but ranks $(1011)^\omega$ and  colors $($\rdl{$a$}\grl{$b$}\yll{$a$}\rdl{$b$}$)^\omega$  
on $P_3$.
The word $(aabb)^\omega$ gets ranks $(1010)^\omega$ and colors 
$($\grl{$a$}\yll{$a$}\grl{$b$}\yll{$b$}$)^\omega$
on $P_1$  but $(1101)^\omega$ 
 and 
$($\rdl{$a$}\rdl{$a$}\grl{$b$}\yll{$b$}$)^\omega$
on $P_2$ and $P_3$.
Our expectation is that the word $(ab)^\omega$ will get rank $1(00)^\omega$ and colors \rdl{$a$}$($\grl{$ba$}$)^\omega$ since every letter, except for the first, closes a required infix $ab$ or $ba$; and the word
 $(aabb)^\omega$ will get rank $11(0101)^\omega$ and colors
\rdl{$aa$}$($\grl{$b$}\yll{$b$}\grl{$a$}\yll{$a$}$)^\omega$
 since, after the first two letters, only every other letter closes such an infix. 
From these, we can compute the robustness value and infer that $(ab)^\omega$ is more robust. 
The DPA $\aut{P}^\truerobustness_{L_{\infty ab}}$ does provide the desired correlation. We can clearly see that the best ranked words have $(ab)^\omega$ as a suffix since the label of the only green cycle is $ab$ (or $ba$), and a green cycle produces the best  value.

The question is then how can we obtain a DPA that has the desired correspondence to our defined notion of $\letterrank_L$. The prominent candidate is the \emph{precise DPA} introduced in~\cite{BohnL24} in the context of passive learning, as it was developed from the notion of $\infixrank_L$ on which $\letterrank_L$ relies, and it shares some of the intuitions of finding reset points. It turns out that it does not give the desired correlation, see~\autoref{ex:precise-aaa} and \autoref{ex:precise-ab}.

Our proof that a DPA with the desired correlation exists has two steps: in the first step we show that a DPA that has the desired correlation between ranks exists, and in the next step we show that a DPA that has the desired correlation between both ranks and colors exists. We term the final DPA the \term{robustness DPA}. The second step is needed only when there are negative ranks. In all other cases the first step suffices. We call the automaton resulting from the first step the \term{vigor DPA} (to avoid something long as \term{pseudo-robustness}).

\vspace{-2mm}
\subsection{Getting the ranks correct: the vigor DPA $\aut{P}^\vigor_L$}\label{sec:vigor-dpa}
\vspace{-1mm}
Let $u,x,y\in\Sigma^*$ and $L\subseteq \Sigma^\omega$.
Let $L_u = \{ v^\omega ~|~ uv \sim_L u.\ uv^\omega \in L  \}$.
We turn to define the vigor equivalence classes, denoted $\equiv^\vigor_u$.
We say that $x\equiv^\vigor_u y$  if there is no suffix $z$ we can read after $x$ and $y$ for which  the ranks of the  dominant suffixes of $xz$ and $yz$ are different.\footnote{
These equivalence classes are reminiscent of the colorful FDFA equivalence classes~\cite{FismanGZ24}.
The difference is that in the colorful FDFA there is no transition from a state ranked $d$ to a state ranked $d'<d$ whereas in the vigor DPA there might be. As per the intro of \autoref{sec:define-rbst}, the idea behind here is to "correct" this non-increasing phenomena 
 by computing the ranks from scratch from time to time (namely resetting). }

\begin{definition}[$\equiv^\vigor_u$]\label{def:robustness}
We say that $x$ and $y$ are \term{vigor-equivalent} wrt $u$ and $L$, denoted  $\mathterm{x \equiv^\vigor_u y}$, if  $\infixrank_{L}(u,\domsuf_{L_u}(xz))= 
     \infixrank_{L}(u,\domsuf_{L_u}(yz))$ for all $z\in\Sigma^*$.
\end{definition}
From the vigor equivalence classes we can construct in the usual manner the respective automaton structure.
Let $\aut{A}^\vigor_u$ be the automaton induced by $\equiv_u^\vigor$. Adding ranks to the states we can obtain a parity automaton $\aut{P}^\vigor_u$.

\begin{definition}[$\aut{P}^\vigor_u$]\label{def:Pu}
    $\mathterm{\aut{P}^\vigor_u}$ is the DPA obtained from $\aut{A}[\equiv_u^\vigor]$ by ranking state $q_x$ by
     $\infixrank_{L}(u,x)$ 
    where $x$ is the smallest word in the shortlex order reaching $q_x$.\footnote{Where $x < y$ by the  \emph{shortlex} order if ${|x|<|y|}$ or ${|x|=|y|}$ and $x$ precedes $y$ in the lexicographical order.}
\end{definition}

\begin{restatable}[]{remark}{rmrkshortlexnoincrese}\label{rmrk:shortlex-no-increse}
If $x$ and $y$ are the smallest words in the shortlex order reaching state $q_x$ and $q_y$ of $\aut{P}_u^\vigor$, resp. and $x$ is a prefix of $y$ then $\kappa(q_x)\geq \kappa(q_y)$.
\end{restatable}

We claim that $\aut{P}^\vigor_u$ recognizes the language $L_u$ and moreover, the ranks of states visited along the way match one to one the natural ranks induced by $L_u$.

\begin{restatable}[Correctness of $\aut{P}^\vigor_u$]{theorem}{thmrecolorateoneequiv}\label{thm:recolorate-one-equiv}
\begin{inparaenum}
\item \label{correctness:lang} $\sema{\aut{P}_u^\vigor}=L_u$ and 
\item \label{correctness:letters} $\letterrank_{L_u}(v^\omega,k)=\kappa(\aut{P}_u^\vigor(v^\omega[..k]))$ for every $v\in\Sigma^+$ and
$k\in\mathbb{N}$. 
\end{inparaenum}
\end{restatable}

The DPAs for $\aut{P}^\vigor_u$ for all equivalence classes of $L$ can be combined into one DPA, 
that recognizes $L$ and its traversed ranks match one-to-one $\letterrank_L$ of the respective letters.

\begin{definition}[$\aut{P}_L^\vigor$]\label{def:rbst-dpa} 
Assume $\sim_L$ has $m$ equivalent classes $c_1,...,c_m$.  
Let $u_1,...,u_m$ be finite words representatives of $c_1,...,c_m$, resp.
Let $\mathterm{\aut{P}_L^{\vigor}}$ be a 
DPA with automaton structure $\aut{A}{[\sim_L]}\times 
\aut{A}[\equiv^\vigor_{u_1}]\times \cdots \times
\aut{A}[\equiv^\vigor_{u_m}]$
and the ranking function 
$\kappa(c_i,q_1,...,q_m)=\kappa_i(q_i)$ where $\kappa_i$ is the ranking function of $\aut{P}^\vigor_{u_i}$. The DPA
$\aut{P}_L^\vigor$ is termed the \term{vigor DPA} for $L$.
\end{definition}

\autoref{fig:P-rbst} provides the vigor DPA for a language with two equivalence classes, particularly for $L_{\text{mod}2}$ from  \autoref{rem:L-mod-2} and \autoref{ex:infix-rank-mod-2}.

\begin{restatable}[Correctness of $\aut{P}^\vigor_L$]{theorem}{thmrecoloratemanyequiv}\label{thm:recolorate-many-equiv}
\begin{inparaenum}
\item \label{correctness:lang} $\sema{\aut{P}^\vigor_L}=L$ and 
\item \label{correctness:letters} $\letterrank_{L}(w,k)=\kappa(\aut{P}^\vigor_L(w[..k]))$ for every lasso word $w\in\Sigma^\omega$ and 
$k\in\mathbb{N}$. 
\end{inparaenum}
\end{restatable}

\begin{restatable}[$\equiv_L^\vigor$]{remark}{rmrkrbstquivclasses}\label{rmrk:rbst-equiv-classes}
$\aut{P}_L^\vigor$'s aut structure corresponds to the equivalence classes $\equiv_L^\vigor$ defined as follows: $x \equiv_L^\vigor y$ if and only if \begin{inparaenum}
    \item
 $x\sim_L y$ and 
 \item 
 for every equivalence class $c$ of $\sim_L$: $x \equiv_c^\vigor y$. 
\end{inparaenum}
\end{restatable}

\subsubsection*{From an arbitrary DPA to the vigor DPA}\label{sec:from-arb-dpa-to-vigor-dpa}

We turn to ask, given an arbitrary DPA $\aut{P}$ for $L$, how can we construct the vigor DPA of $L$, $\aut{P}_L^\vigor$. Following \autoref{rmrk:rbst-equiv-classes} this amounts to finding the equivalence classes of
$\equiv_L^\vigor$.

To do so, we note first that we can find the equivalence classes of $\sim_L$ by checking for every pair of states $q_1,q_2$ of $\aut{P}$ whether $\sema{\aut{P}_{q_1}}=\sema{\aut{P}_{q_2}}$ where $\aut{P}_q$ is the DPA obtained from $\aut{P}$ by making $q$ the initial state. DPA equivalence can be checked in polynomial time~\cite{Schewe10,AngluinF24}.

Second, we note that there exists polynomial time procedures that lower the rank of each state of a given DPA to the minimal rank possible (that does not change the language)~\cite{NiwinskiW98,CartonM99}. Next we show that from such a DPA  $\aut{P}$ we can compute $\infixrank_L(u,v)$ and the $\predomsuf_{L_u}(v)$.

\begin{restatable}[]{lemma}{lempotentiallineartime}\label{lem:potential-linear-time}
    Given a DPA $\mathcal{P}$ and words $u\in\Sigma^*$, $v\in\Sigma^+$, we can compute $\infixrank_L(u,v)$ and 
    $\predomsuf_{L_u}(v)$ in  time polynomial in $|uv|+|\aut{P}|$.
\end{restatable}

The proof shows how to compute $\infixrank_L(u,v)$, from which we can compute whether index $k$ of $w$ is influential, and in turn the predominant suffix of a given finite word. 

\begin{restatable}[]{lemma}{lemcomputerbstpfromp}\label{lem:compute-rbst-p-from-p}
Given a DPA $\aut{P}$ for $L$, we can compute $\aut{P}^\vigor_L$ in time polynomial in $|\aut{P}|+|\aut{P}^\vigor_L|$.
\end{restatable}

The proof uses ideas of active learning~\cite{Angluin87} to compute the equivalence classes of $\aut{P}_L^\vigor$. In a sense, the dominant suffix equivalence classes are  computed from queries on the predominant-suffixes.

\begin{remark}
  We note that while this procedure is polynomial in the size of the given input DPA $\aut{P}$ and the vigor DPA $\aut{P}^\vigor_L$, the size of
 $\aut{P}^\vigor_L$ can be exponential in the size of $\aut{P}$ as is the case for the precise DPA for a language~\cite{BohnL24}.  A family of languages exhibiting this is $L_n$ over $\Sigma_n=\{a_1,...,a_n\}$ requiring each letter of $\Sigma_n$ to occur infinitely often. 
 There is a DPA for $L_n$  with $n$ states, yet every {subset} of $\Sigma_n$ resides in its own equivalence class in $\equiv_{L_n}^\vigor$. 
 \footnote{A DPA for $L_n$ that needs not give the correct color for every read prefix, can decide on an arbitrary order of elements in $\Sigma_n$ and after reading the $i$-th letter wait for the $(i{+}1)$-th letter (modulo $n$). Whereas a DPA that has to provide the correct color on every read prefix, has to signal that all letters of $\Sigma_n$ have arrived as soon as this is the case (and thus is required to record all subsets of $\Sigma_n$).}
\end{remark}

\subsection{Getting the colors correct: the robustness DPA $\aut{P}^\truerobustness_L$}\label{sec:truerobusness-dpa}
The vigor DPA $\aut{P}^\vigor_L$ from~\autoref{def:rbst-dpa} provides the desired tight correlation between the ranks of the traversed edges and the ranks given to letters by $\letterrank_L$.
We can color the edges of the resulting DPA using $\theta(d,d')$ from \autoref{def:letter-color} where $d$ and $d'$ are the ranks of the edges' source and target. This, however, will induce a tight correlation to the color of the letters as per \autoref{def:letter-color} only 
in case there are no black or white letters. 
If such exist, then traversing a black (resp. white) edge will result in reaching a rejecting (resp. accepting) terminal MSCC of the DPA, and thus all subsequent transitions will be black (resp. white). Since the color of letters (\autoref{def:letter-color}) builds on the forgetful version of the read prefix, there could be non-black (resp. non-white) letters after a black (resp. white) letter.  \autoref{ex:black-edges-issue} exemplifies this issue. To remedy this we introduce a new type of acceptance for $\omega$-automata: \emph{forgetful parity}.

\begin{definition}[Forgetful parity automaton]\label{def:forgetful-parity}
    A \term{forgetful parity} automaton is a tuple $(\Sigma,Q,q_0,\delta,\kappa)$ like a parity automaton, with three differences. First, the minimal rank is $-2$.
Second,
there are $\varepsilon$-transitions. We require that if $\delta(q,\varepsilon)$ is defined then $\delta(q,\sigma)$ is undefined for every $\sigma\in\Sigma$. Thus, the automaton is deterministic. 
Third, a run $\rho=q_0q_1q_2...$ of a forgetful parity automaton is \emph{accepting} if either \textcolor{lipicsBulletGray}{\textbf{(*)}} there exists  $i\in\mathbb{N}$ such that $\kappa(q_i)=-2$ and $\kappa(q_j) \geq 0$ for all $j<i$ or \textcolor{lipicsBulletGray}{\textbf{(**)}} $\kappa(q_i)\geq 0$ for all $i\in\mathbb{N}$ and 
$\min(\inf(\{\kappa(q_i)~|~i\in\mathbb{N}\})$ is even.
\end{definition}

Note that (**) is the usual parity acceptance condition, and (*) reflects that the word reached a terminal accepting MSCC.

Given a forgetful parity automaton, we color its edges as follow: edge $(q,\sigma,q')$ is colored  $\theta(\kappa(q),\kappa(q'))$ using $\theta$ of \autoref{def:letter-color}. 
An $\varepsilon$-edge is not colored. See \autoref{fig:aut-for-examples-parity-colored}.
Let  $w=\sigma_1\sigma_2... \in\Sigma^\omega$ and let $\rho_w=q_0q_1q_2...$ be the run of a forgetful parity automaton $\aut{P}$ on $w$.
We use $\mathterm{\ecolor^\aut{P}_{w}(i)}$ for the color of the edge $(q_{i-1},\sigma_i,q_i)$. 

\begin{restatable}[The robustness DPA]{proposition}{propedgecolorcorrectness}\label{prop:edge-color-correctness} Let $\aut{P}_L^\vigor$ be the vigor DPA for $L$. There exists a polynomial time procedure returning a forgetful DPA $\mathterm{\aut{P}^\truerobustness_L}$ such that $\sema{\aut{P}_L^\vigor}=\sema{\aut{P}^\truerobustness_L}$ and 
$\ecolor^{\aut{P}^\truerobustness_L}_w(i)=\lcolor_w({i})$
for all $w{\in}\Sigma^\omega$ and $i{\in}\mathbb{N}$.
We term $\mathterm{\aut{P}^\truerobustness_L}$ the \term{robustness DPA} for $L$. 
\end{restatable}

The proof uses a simple procedure that replaces $(q,\sigma,q')$ transitions to state $q'$ in a terminal accepting MSCCs by $(q,\sigma,q_{\top})$-transition to a dedicated state $q_\top$, and a respective $\varepsilon$-transition $(q_\top,\varepsilon,q)$, and similarly for transitions to terminal rejecting MSCCs.
See \autoref{fig:no-redundant-grants} and \autoref{fig:danger} for the vigor and robustness DPAs for the languages in \autoref{ex:no-redundant-grants} and \autoref{ex:danger}.

\subsection{Computing robustness values}

The robustness value of a given word can be computed in polynomial time from the robustness automaton, as we show next. Following this, given two words we can compute which one is preferred according to the robustness preference relation.

\begin{restatable}[]{proposition}{propcomputevalue}\label{prop:compute-value}
    Given the robustness automaton $\aut{P}_L^\truerobustness$ for $L$ and a lasso representation $(u,v)$, we can compute $\valrbst_{L}(uv^\omega)$ in polynomial time.
\end{restatable}

The proof relies on the following claim correlating $\rbstdecomp_L()$ to the robustness DPA.

\begin{restatable}[]{myclaim}{clmrbstdcmptoautdcmp}\label{clm:rbst-dcmp-to-aut-dcmp}
Let $\aut{P}^\truerobustness_L$ be the robustness DPA for $L$ and let $w$ be a lasso word. Let $q$  be the first state reached in the final loop of the run of $\aut{P}^\truerobustness_L$ on $w$. Let $u$ be the prefix of $w$ read from the initial state to $q$ and $v$ the infix of $w$ read from $q$ to $q$. Then $\rbstdecomp_L(w)=(u,v)$.
\end{restatable}

\section{Conclusions}\label{sec:discussion}
We have shown that we can distill from a given $\omega$-regular language $L$ a robustness preference relation
that refines Tabuada and Neider's idea into an infinite domain,
 meets our intuitions in all provided examples, and satisfies the desired criteria of duality, value wrt the complement, and satisfaction bar. Towards this aim we provided semantical definitions for the natural rank and  color of a letter of an $\omega$-word wrt $L$. 
 From these we have constructed the \emph{robustness DPA} for $L$, a DPA whose equivalence classes capture the semantical notion of robustness; and has a one-to-one correspondence between the ranks and colors of the edges traversed when reading the word and the ranks and colors of the letters of the word. For future we intend to show
 that these ideas can be used to construct a runtime consultant that provides recommendations for next actions in order to achieve executions as robust as possible given the history and desired spec.

\subparagraph*{Acknowledgments}
We thank Matan Pinkas and Oded Zimerman for comments on an early draft of this paper.

\bibliographystyle{plainurl}
\bibliography{bib.bib}

\appendix

\section{Omitted proofs}

\lemds*
\begin{proof}
    Since the definition of the reset points is inductive we prove the claim using induction. For the base case $k{=}0$ the only option for the reset point is $k$ itself, and so the claim holds. For the induction step, let $k\in\mathbb{N}$. If $k$ is influential wrt $\reset_w(k{-}1)$ and $j_k$ is the  index witnessing this, then $j_k$ also witnesses that $k$ is influential wrt any  $i<\reset_w(k{-}1)$.
    
    Otherwise, $k$ is not influential wrt $\reset_w(k{-}1)$. 
    First we show that if $k$ is not influential wrt $\reset_w(k{-}1)$ then $k$ is not influential wrt any $i\leq \reset_w(k{-}1)$. Assume towards contradiction that there exists $i\leq \reset_w(k{-}1)$ such that $k$ is influential wrt $i$. 
    Then there exists $j\in[i..k]$ such that $\infixrank_L(w,j,k)<\infixrank_L(w,j,k{-}1)$. Since $k$ is not influential wrt $\reset_w(k{-}1)$ then $j<\reset_w(k{-}1)$. 
    Thus $\infixrank_L(w,j,k)<\infixrank_L(w,\reset_w(k{-}1),k)$. 
    Hence,  $w[k]$ is a word  distinguishing $w[j..k{-}1]$ from $w[\reset_w(k{-}1)..k{-}1]$, contradicting that $\reset_w(k{-}1)$ is the dominant-index of $w[..k{-}1]$ (as per
    \autoref{def:dom-ind}).
    
    If $k{-}1$ is influential then {$\infixrank_L(w,k,k)>\infixrank_L(w,\reset_w(k{-}1), k{-}1)$}
    because there is necessarily a decrease in the infix rank when reading $w[k{-}1]$. 
    That is, the dominant index for $k$ is 
    greater than $\reset_w(k{-}1)$ and it holds wrt to any $i<\reset_w(k{-}1)$. 
    
    Otherwise, both $k$ and $k{-}1$ are not influential wrt $\reset_w(k{-}1)$ and $\reset_w(k{-}2)$ resp. Then, 
    {$w[\predomind(w,0,k{-}1)..k{-}1]$} is the longest suffix of $w[..k{-}1]$ with the same infix rank as 
    $\infixrank_L(w,k{-}1,k{-}1)$.
    Any suffix starting before 
    {$\predomind(w,0,k{-}1)$} would decrease the infix rank. Also, 
    $\infixrank_L(w,k,k)\geq\infixrank(w,k{-}1,k{-}1)$
    because $k$ is not influential. 
    Hence, the predominant index {of $k$} 
    is at most $\reset_w(k{-}1)$. {By \autoref{def:dom-ind},} the dominant index of $k$ {corresponds to the shortest suffix of 
    {$\predomind(w,0,k{-}1)$} that follows the requirements of \autoref{def:dom-ind} so the dominant index is at least $\reset_w(k{-}1)$. Therefore, the dominant index of $k$} wrt any index $i\leq \reset_w(k{-}1)$ is the same.
\end{proof}

\commentout{
\begin{definition}[Determining index, determining suffix]\label{def:det-suf}
Let $w\in\Sigma^*\cup\Sigma^\omega$.
Let $i \leq k \leq|w|$. Let $w[l..k]$ be the dominant suffix of $w[i..k]$.
Let $w[j..k]$ be the \emph{shortest} suffix of $w[l..k]$ for which $\infixrank_L(w,j,k)=\infixrank_L(w,l,k)$ and $\infixrank_L(w[..j),w[j..k]\cdot v )=\infixrank_L(w[..l),w[l..k]\cdot v)$ for all $v\in\Sigma^*$. Then $w[j..k]$ is the \term{determining suffix} of $w[i..k]$, denoted $\mathterm{\detsuf_L(w,i,k)}$, and $j$ is the \term{determining index} of $w$ wrt $i,k$ and $L$, denoted $\mathterm{\detind_L(w,i,k)}$.
\end{definition}
}

The following two claims are needed for subsequent proofs.

\begin{myclaim}\label{clm:if-only-non-inlf-indices-then-max-rank}
Let $L\subseteq \Sigma^\omega$, $m$ be the maximal rank in $L$,
and $w\in\Sigma^\omega$. If there exists $n_0\in\mathbb{N}$ such that for all $n>n_0$ the $n$-th index of $w$ is non-influential then $\letterrank(w,n)=m$ for all $n>n_0$. 
\end{myclaim}

\begin{proof}
Assume towards contradiction that $\letterrank(w,n)=m'$ for some $m'<m$. Then, since $n$ is a non-influential index, by \autoref{def:influential}, \autoref{def:dom-ind} and \autoref{def:letter-rank}, $\infixrank_L(w,\reset_w(n),n)=\infixrank_L(w,n,n)=m'$. This means that $\infixrank_L(w[..n),w[n])<\infixrank(w[..n),\varepsilon)$. 
Thus, the $n$-th letter is influential. Contradiction.
\end{proof}

\begin{myclaim}\label{clm:reading-not-from-epsilon-decreasing-rank}
Let $L\subseteq \Sigma^\omega$ and $x\in\Sigma^*$. 
If $\letterrank_L(x,|x|)=d$ then $\letterrank_L(zx,|zx|)\leq d$ for all $z\in \Sigma^*$.
\end{myclaim}

\begin{proof}
    Let $z\in \Sigma^*$ be some finite word. We claim that since $x$ is a suffix of $zx$ then $\domsuf_L(x)$ is a suffix of $\domsuf_L(zx)$. 
    Assume this is not the case and let $|z|=k$ and $|x|=l$, then $\reset_{zx}(k+l)=m>k$ and 
    $\reset_{x}(l)<m-k$.
    By \autoref{cor:domind} $\domind_L(zx,m,k+l)=\domind_L(zx,m',k+l)$ for every $m'\in[0..m]$ and in particular $\domind_L(zx,m,k+l)=\domind_L(zx,k,k+l)$. Then, $\domind_L(zx,m,k+l)=\domind_L(x,0,l)$ and $\reset_x(l)=m-k$, contradicting that $\reset_{x}(l)<m-k$. Hence, $\letterrank_L(zx,k+l)\leq \letterrank_L(x,l)=d$.
\end{proof}

\thmnatrankoflettersandwords*
\begin{proof}
We show that $\min \inf (n_1,n_2,...) = d$ implies $\wordrank_L(w)=d$. Since the range of $d$ is finite this entails the other direction as well. 
The proof is by induction on $d$.
For the base cases we consider $d\in\{-2,-1,0,1\}$.\footnote{Note that $1$ can be obtained either from the induction base or the induction step.}
If $\min \inf (n_1,n_2,...) = -2$ then there exists at least one natural number $n$ such that $\letterrank_L(w,n)=-2$. It follows that $\infixrank_L(w,\reset_w(n),n)=-2$.
Thus, by \autoref{def:infix-rank} for all $w'\in\Sigma^\omega$ we have that $w[..n]w'\in L$. Hence $\wordrank_L(w)=-2$.
The argument for $-1$ is similar. 

Otherwise, no negative rank is obtained by reading $w$. 
If $\min \inf (n_1,n_2,...) = 0$ then there exists $n_0\in\mathbb{N}$ such that $\letterrank_L(w,n)\geq 0$ for all $n\geq n_0$ and in particular there exists an infinite sequence  $I=i_1,i_2,...$ for which $\letterrank_L(w,i_j)=0$ for all $j\in\mathbb{N}$. By \autoref{def:letter-rank}, $\letterrank_L(w,i_j)=\infixrank_L(w,\reset_w(i_j),i_j)=0$. 
Let $r_j=\reset_w(i_j)$, then
by \autoref{def:infix-rank}, it holds that {$\infixrank_L(w[..r_j),w[r_j..i_j])=\max\{\wordrank_L(w[..r_j)(w[r_j..i_j]z)^\omega)~|~w[..i_j]z\sim_L w[..r_j)\}$.}
Since the maximal possible word rank is $0$, and $0$ is the minimal possible non-negative rank as well, then $\wordrank_L(w)=0$. 
The base case for $\min \inf (n_1,n_2,...) = 1$ is similar. 

For the induction step assume $\min \inf (n_1,n_2,...) = d$. Then there exists $n_0\in\mathbb{N}$ such that $\letterrank_L(w,n)\geq d$ for all $n\geq n_0$ and in particular there exists an infinite sequence $I=i_1,i_2,...$ such that $\letterrank_L(w,i_j)=d$ for all $j\in\mathbb{N}$. Since $\letterrank_L(w,i_j)=d$ then $\infixrank_L(w,\reset_w(i_j),i_j)=d$ for all $j\in\mathbb{N}$. 
Let $J=j_1,j_2,...\subseteq I$ be such that $\reset_w(j_{n{+}1})\geq j_n$.
We distinguish two cases: (a) $J$ is infinite (b) $J$ is finite.

Assume  $J$ is infinite.
Let $w'\in \inj_J(w)$, then $w'=u_1z_1u_2z_2\ldots$ 
for some $z_1,z_2,\ldots\in\Sigma^+$ such that $u_1z_1 \sim_L u_1$, $u_1z_1u_2z_2 \sim_L u_1z_1u_2$ and so on, and $\infixrank_L(\varepsilon,u_1)=d$, $\infixrank_L(u_1,u_2)=d$, $\infixrank_L(u_1u_2,u_3)=d$ and so on. 
From \autoref{rmrk:infix-no-increse} it holds that $\infixrank_L(\varepsilon,u_1z_1)=d'\leq d$ and $\infixrank_L(\varepsilon,u_1z_1')\in [d'..d]$ for every $z_1' \preceq z_1$ and similarly $\infixrank_L(u_1z_1,u_2z_2)=d''\leq d$ and $\infixrank_L(u_1z_1,u_2z_2')\in [d''..d]$ for every $z_2' \preceq z_2$ and so on.
Let $i_n=|u_1u_2\ldots u_n|$ and $i'_n=|u_1z_1\ldots u_{n-1}z_{n-1}u_n|$. By \autoref{clm:reading-not-from-epsilon-decreasing-rank}, $\letterrank_L(w',i'_n)\leq\letterrank_L(w,i_n)$ for all $n\in\mathbb{N}$. 
We get that 
{$\min\inf\{ \letterrank_L(w',n)\}_{n\in\mathbb{N}}\in[d^\star..d]$ where $d^\star$ is the minimal among $d',d''...$}.
Hence, for every word $w'$ in $\inj_{I}(w)$ we have that either $\min\inf\{ \letterrank_L(w',n)\}\in[d^*..d)$ or $\min\inf\{ \letterrank_L(w',n)\}=d$.
In the first case, by the induction hypothesis, $\wordrank_L(w')<d$. In the second case 
$w\in L$ iff $w'\in L$ from the parity of $d$. Thus, by \autoref{def:inf-nat-rank-minuses}, $\wordrank_L(w)=d$.

Assume $J$ is finite. 
{This} means all the indices from some point have the same minimal rank and are not influential. According to \autoref{clm:if-only-non-inlf-indices-then-max-rank}, their rank is the maximal rank in $L$. 
In this case for all $w'\in \inj_J(w)$ we have $\wordrank_L(w')\leq d$ (since $d$ is maximal). Hence, by \autoref{def:inf-nat-rank-minuses}, $\wordrank_L(w)=d$.
\end{proof}

\rmrkshortlexnoincrese*
\begin{proof}
    Since $x$ is a prefix of $y$ we have $\infixrank(u,y)\leq\infixrank(u,x)$.
    Also, $\kappa(q_y)=\infixrank(u,y)$ and $\kappa(q_x)=\infixrank(u,x)$, as per \autoref{def:Pu}. Then $\kappa(q_y)=\infixrank(u,y)\leq\infixrank(u,x)=\kappa(q_x)$.
\end{proof}

\thmrecolorateoneequiv*
\begin{proof}
We show the second item.  
The first item  
follows from the second and \autoref{thm:natrank-of-letters-and-words}.

We name the state of $\aut{A}[{\sim}]$ corresponding to class $[u]$ by $q_x$ where $x$ is the least word in the shortlex order that resides in $[u]$. 
Since $\aut{A}^\vigor_u$ is the automaton structure induced from $\equiv_u^\vigor$, by definition for every word $y$ if on reading $y$ the automaton reaches state $q_x$ then $y \equiv_u^\vigor x$.
Assume $w[..k]$ reaches state $q_x$. Thus, $w[..k]\equiv_u^\vigor x$.
It follows from \autoref{def:robustness} that for all $z\in\Sigma^*$, in particular $z=\varepsilon$ we have $\infixrank_{L_u}(u,\domsuf_{L_u}(x))=\infixrank_{L_u}(u,\domsuf_{L_u}(w[..k])$.
Note that the LHS is $\kappa(q_x)=\kappa(\aut{P}^\vigor_u(w[..k])$ and the RHS is $\letterrank_{L_u}(w,k)$. Thus, the claim follows.
\end{proof}

\thmrecoloratemanyequiv*
\begin{proof}
We show the second item.  
The first item  
follows from the second and \autoref{thm:natrank-of-letters-and-words}.

We use $\inf_{\sim_L}(w)$ for the set 
$\{[u]_{\sim_L}~|~ \forall i. \exists j>i.\ w[..j]\in [u]\}.$
Suppose $w=uv^\omega$, $u\sim_L uv$, $\wordrank_L(w)=d$ and  $\inf_{\sim_L}(w)=\{u_1,...,u_l\}$ and $m\geq l$ is the number of equivalence classes of $\sim_L$.
Let $y$ be a shortest infix of some repetition of $v$ that suffices to conclude the rank of $uv^\omega$.
That is,  $v^n=v'v''$ for some $n\in\mathbb{N}$, $v'v''\in\Sigma^*$ and $v''=xy$ for some $x,y\in\Sigma^*$ where $y=\domsuf_L(v'')$.\footnote{E.g., for $L_{a\text{-seq}}$, $u=b$, $v=bbaa$ we have $y=aa$. For $L_{a\text{-seq}}$, $u=b$, $v=a$ we have $y=aaa$.}
Then $\infixrank_L(uv'x,y)=d$. 
Let $u_i$ be such that $uv'xy_i'\sim_L u_i$ for some  $y_i'y_i''=y$. Then $\infixrank_L(uv'xy_i',y_i''y)=d$.
It follows from~\autoref{thm:recolorate-one-equiv} that $\kappa_{i}(\aut{P}_{u_i}(y_i''y))=d$ for every $i\in[1..l]$. 

Let $\rho_w=q_0,q_1,...$ be the run of $\aut{P}^\vigor_L$ on $w$ where $q_i=(c_i,p_{i,1},p_{i,2},...,p_{i,m})$. 
Then for every $i\in\mathbb{N}$ such that $y$ is a suffix of $w[..i]$ we have that $\kappa_{j}(p_{i,j})=d$ for all $j\in[1..l]$ and for infinitely many such $i's$ we have that $c_i\in[1..l]$.
Thus, $\kappa(\aut{P}^\vigor_L(w[..i]))=\kappa_{c_i}(p_{i,c_i})=d$ for infinitely many $i'$s. 
Recall that we have already established in \autoref{thm:recolorate-one-equiv} that for some $n_0\in\mathbb{N}$ it holds that
$\kappa_{j}(p_{n,j})\geq d$
for every $n>n_0$ and $j\in[1..l]$. 
It follows that the minimal rank visited infinitely often by $\aut{P}$ when running on $w$ is $\wordrank(w)$.
\end{proof}

\lempotentiallineartime*
\begin{proof}
We start with the computation of $\infixrank_L(u,v)$.

First we compute the MSCCs of $\aut{P}$.
     A state is termed \emph{transient} if it is not part of any MSCC (recall that a singleton $\{q\}$ is an MSCC only if there is a self-loop on $q$).
    Then we classify the MSCCs into \emph{rejecting}, \emph{accepting} or \emph{mixed}. An MSCC is \emph{rejecting} (resp. \emph{accepting}) if all its ranks are odd (resp. even) and it is \emph{mixed} otherwise. A rejecting (resp. accepting) MSCC is \emph{terminal} if only rejecting (resp. accepting) MSCCs are reachable from it. 
    
    Let $q_u$ and $q_{uv}$ be the states reached on reading $u$ and $uv$, resp. from the initial state.
If $q_{uv}$ is transient return $\infty$, if $q_{uv}$ resides in a rejecting (resp. accepting) terminal MSCC return $-1$ (resp. $-2)$.

    Otherwise, let $d$ be the minimal rank visited when reading $v$ starting from $q_u$.
    Remove all states with rank lower than $d$ and check if there is a path from $q_{uv}$ to $q_u$ (can be done by a BFS). If so, $\infixrank_L(u,v)=d$. Otherwise, decrease $d$ by $1$ and perform the same for the new value. The 
    runtime is $O(m|E|{+}|uv|)$ where $m$ is the number of different ranks in the automaton, and $E$ are the edges of $\aut{P}$.

    For $\predomsuf_{L_u}(v)$, let $l=|v|$. 
    We compute first $d_1,d_2,d_3,... ,d_{l}$ and $d'_1,d'_2,d'_3,... ,d'_{l}$ where 
    $d_i=\infixrank(u,v[i..])$ and $d'_i=\infixrank(u,v[i..))$ for $i\in[1..l]$.
    If there is an $i$ for which $d_i<d'_i$, then $l$ is influential.
    Once we know whether $l$ is influential, we conclude whether $\predomsuf_{L_u}(v)$ should be the longest or the shortest suffix available, as defined in \autoref{def:dom-ind}. If index $l$ is not influential then the predominant index is the smallest $j$ such that $d_j=d_l$. 
    Otherwise, the predominant index is the biggest $j$ such that $d_j=d_1$.
    Given $j$ is the predominant index, the predominant suffix $\predomsuf_{L_u}(v)$ is $v[j..]$.
\end{proof}

\lemcomputerbstpfromp*
\begin{proof}
First we find the equivalence classes of $\sema{\aut{P}}$ as discussed above. Then for each equivalence class $u$, we can compute $\equiv_u^\vigor$, using an active learning algorithm such as $L^\star$~\cite{Angluin87} where the observation table entry $(x,y)$ is filled with the results of {$\infixrank_L(u,\predomsuf_{L_u}(xy))$} 
that by \autoref{lem:potential-linear-time} can be answered in polynomial time. Equivalence queries are answered via the polynomial time equivalence algorithm~\cite{AngluinF24} for $\aut{P}_{q_u}$ where $q_u$ is the state of $\aut{P}$ reached by reading $u$. This will result in $\aut{P}_u^\vigor$. From $\aut{P}_{u_i}^\vigor$ of all equivalence classes $u_i$ we can construct $\aut{P}^\vigor$ following \autoref{def:rbst-dpa}. 
\end{proof}

\propedgecolorcorrectness*
\begin{proof}
We modify $\aut{P}^\vigor$ by removing $-2$-terminal and $-1$-terminal MSCCs, and introducing states $q_\top$ ranked $-2$ (resp. $q_\bot$ ranked $-1$) for every state $q$ with a transition to $-2$-terminal (resp. $-1$-terminal) MSCC. 
Then, remove every edge $(q,\sigma,q')$ where $q'$ is in a $-2$-terminal (resp. $-1$-terminal) MSCC and replace it by an edge $(q,\sigma,q_\top)$ (resp. $(q,\sigma,q_\bot)$) and an edge $(q_\top,\varepsilon,q)$ (resp. $(q_\bot,\varepsilon,q))$. 
That is, for each such state $q$ its own dedicated state $q_\top$ (resp. $q_\bot$) is added.
 According to \autoref{def:letter-color}, reading the letter $\sigma$ by going through edge $(q,\sigma,q_{\top})$ (resp. $(q,\sigma,q_\bot
)$) gets the white (resp. black) letter color, because the $q_\top$ (resp. $q_\bot$) represents the accepting (resp. rejecting) terminal MSCCs. Therefore, the color of the corresponding edge is white (resp. black). Let $\aut{P}^\truerobustness$ be the resulting forgetful parity automaton.

For the equality $\sema{\aut{P}^\vigor}=\sema{\aut{P}^\truerobustness}$, let  $w\in\Sigma^\omega$, and let
$\rho_w=q_0q_1q_2...$ be the run of $\aut{P^\vigor}$ on $w$ and $\rho'_w=q'_0q'_1q'_2...$ be the run of $\aut{P}^\truerobustness$ on $w$.
If a negative rank is never traversed by $\aut{P}^\vigor$ then
$\kappa(q_i)=\kappa(q'_i)$ for every $i$ and so $w$ is accepted by both $\aut{P}^\vigor$ and $\aut{P}^\truerobustness$ if the minimum rank visited infinitely often is even.
If a negative rank is traversed by $\aut{P}^\vigor$, then if the first such rank is $-2$, then for some $j$ we have $\kappa(q_j)=-2$ and $\kappa(q_i)\geq 0$ for all $i<j$.
In this case $\kappa(q_k)=-2$ for all $k>j$ and so $w\in\sema{\aut{P}^\vigor}$.
Since in $\aut{P}^\truerobustness$ we also have $\kappa'(q'_j)=-2$ and $\kappa'(q'_i)\geq 0$ for all $i<j$,
by the definition of a forgetful DPA, $w\in\sema{\aut{P}^\truerobustness}$.
Similarly, if the first negative rank is $-1$ then it is $-1$ all along in $\aut{P}^\vigor$ and so $w\notin\sema{\aut{P}^\vigor}$ and it is not in $\sema{\aut{P}^\truerobustness}$ by the definition of a forgetful parity automaton.

Since  $\aut{P}^\vigor$ has the correct ranks, it has the correct colors of transitions when black/white edges are not involved.
In $\aut{P}^\truerobustness$ also runs with black/white edges are colored correctly since the transition to the dedicated state $q_\top$ (resp. $q_\bot$) colors the edge correctly with white (resp. black) and then the $\varepsilon$-transition mimics the forgetfulness by going back to the same state and allowing all colors to be read from there.
\end{proof}

\propcomputevalue*

\begin{proof} 
    Following \autoref{clm:rbst-dcmp-to-aut-dcmp} we can compute the robustness decomposition $(x,y)$ of $w=uv^\omega$ in polynomial time, by finding the first state $q$ the word loops at. 
    By \autoref{prop:edge-color-correctness} there is a one-to-one correlation between the colors of the infixes $x$ and $y$ and the colors of the edges traversed when reading them in $\aut{P}_L^\truerobustness$. Applying \autoref{def:infix-scr}
    we can compute the corresponding average scores $\tau_x$ and $\tau_y$.   
    We note that $w\in L$ if either a white edge is traversed in reading $xy$ (and no black edge is traversed before) or the minimal rank  visited when reading $y$ from $q$ is even. We set $a_w$ accordingly. 
    By \autoref{def:rbst-scr} the robustness value of $w$ is $(a_w,\tau_y,|x|(\tau_x-\tau_y))$. All the procedures described above can be done in time polynomial in {$|uv|{+}|\aut{P}_L^\truerobustness|$}.
\end{proof}

\clmrbstdcmptoautdcmp*
\begin{proof}
    We have to show that  $(u,v)$ is the shortest decomposition satisfying $w=uv^\omega$, $\infixrank_L(u,v)=\wordrank_L(w)$.
    Let $d$ be the minimal rank on the $v$-loop from $q$.
    By \autoref{rmrk:rbst-equiv-classes} we have that $uv\sim_L u$.
    By \autoref{thm:recolorate-many-equiv} if $\rho_w=q_0q_1q_2...$ is the run of 
    $\aut{P}^\truerobustness_L$  on $w$, then $\kappa(q_i)=\letterrank_L(w,i)$ for every $i\in\mathbb{N}$.
    It follows that $\min (\inf \{\letterrank_L(w,i)~|~i\in\mathbb{N}\})=d$. By \autoref{thm:natrank-of-letters-and-words} $\wordrank_L(w)=d$.
    
    We show that $\infixrank_L(u,v)=d$.
    Assume this is not the case, then either (i) $\infixrank_L(u,v)=d'>d$ or (ii) $\infixrank_L(u,v)=d'<d$.
    In the first case, by \autoref{def:infix-rank} there exists $z\in\Sigma^*$ such that $\wordrank(u(vz)^\omega)=d'$ and $uvz\sim_L u$.
    Note that the run of $\aut{P}^\truerobustness_L$ on $u(vz)^\omega$ visits state $q$ with rank $d<d'$. This contradicts \autoref{thm:natrank-of-letters-and-words}.  In the second case, by \autoref{def:infix-rank} for all words $z\in\Sigma^*$ such that $uvz\sim_L u$ we have that $\wordrank_L(u(vz)^\omega)\leq d'$.
    But for $z=\epsilon$ we know $\wordrank_L(u(v\varepsilon)^\omega)=d>d'$. Contradiction.

    To see that it is the shortest such decomposition, let
     $|u|=l$ and $q_1,q_2,...,q_l$ be the sequence of states obtained by reading $u$ from $q_\epsilon$. Note that $q_l=q$.
    Assume towards contradiction that there exists a shorter decomposition $(u',v')$ of $w$. Since $w$ is a lasso word and there is one final loop of $w$ on $\aut{P}_L^\truerobustness$, it holds that the period $v'$ is some rotation of $v$. Then $u'\prec u$ and there exists $q_i$ for some $i\in[1..l)$ that reaches the final loop before $q$ while reading $u$. This contradicts that $q$ is the first state of the final loop that is reached.
\end{proof}

\section{Additional examples and figures}

\subsection{Additional motivating examples}

\begin{example}[Response time, Starting time]\label{ex:responset_time}
    A common LTL property is $L=\allowbreak\ltlG(r \to \ltlF g)$ stipulating that every request (signal $r$) should eventually be followed by a grant (signal $g$). 
    Assume the alphabet $\Sigma=2^{\{r,g\}}=\{rg,r\overline{g},\overline{r}g,\overline{rg}\}$.
    Wrt $L$, many papers~\cite{BloemCHJ09,BloemGHJ09,BloemCGHHJKK14,AlmagorBK16,BartocciBNR18,AlmagorK20,HenzingerS21,HenzingerMS22} 
    view $(r\overline{g}\cdot (\overline{rg})^i \cdot \overline{r}g)^\omega$ as more robust than $(r\overline{g} \cdot (\overline{rg})^j \cdot \overline{r}g)^\omega$ if $i<j$. That is, we view a trace where the response time is shorter as more robust compared to one where it is longer. 

    Wrt $L'=\signal{reset}^*\cdot \signal{start}\cdot (\signal{busy}^*\cdot \signal{done})^\omega$ for $i<j$ we require (using first letter abbreviations of the signals) 
     $(r)^i  s (b^k  d)^\omega \gtrbst_{L'}
    (r)^j s  (b^k  d)^\omega$
    where $\gtrbst_L$ stands for \emph{more robust wrt $L$}. That is, we view a trace where the \signal{start} signal arrives earlier as more robust.
\end{example}

\begin{example}[No redundant grants]\label{ex:no-redundant-grants}
    In the context of synthesis, it is observed that the above property $\ltlG(r \to \ltlF g)$ is not strong enough as it puts no requirements on issuing grants~\cite{BloemCHJ09,BloemCGHHJKK14}. In real implementations issuing grants is costly and often one strengthens this property by requiring no redundant grants. The resulting property is $L''=\ltlG(r \to \ltlF g) \wedge \ltlG (g \to \ltlX [\neg g \ltlW r]) \wedge [\neg g \ltlW r]$. 
    Here 
    we would like $(r\overline{g}\cdot \overline{rg}\cdot \overline{r}g)^\omega \gtrbst_{L''} (r\overline{g}\cdot \overline{rg}\cdot  (\overline{r}g)^2)^\omega \gtrbst_{L''} (r\overline{g}\cdot \overline{rg}\cdot (\overline{r}g)^i)^\omega$ for $i>2$. That is, we prefer traces with fewer redundant grants.
\end{example}

\begin{example}[Avoid debts]\label{ex:avoid-debts}
Consider the language  $L=a\Sigma^\omega \vee bc^\omega$. 
If the first letter is $a$ then there are no further requirements. If the first letter is $b$ then there is a constant demand that $c$ will hold in every future tick. Thus any word starting with $a$ should be preferred over any word starting with $b$, which in a sense constantly creates a debt that should be satisfied in the future. This example together with $L_{a\text{-seq}}$ from \autoref{ex:diff-priorities-for-red-yellow-green-cycles} illustrate the requirement that words that create less debts should be preferred over words that create more debts.
\end{example}

\subsection{Additions to aid understanding the technical details}

\begin{remark}[Injection Sequences]\label{remark-exists-I}
   Note that \autoref{def:inf-nat-rank-minuses} uses 
    the \emph{existence} of a sequence $I$ of injection points.    
    For a given language $L$ and given word $w$ there could be sequences $I'$ and $I''$ such that $I'$ can be used in the definition to prove that the natural rank of $w$ wrt $L$ is what it is, while $I''$ cannot. Consider for instance $L_{a\text{-seq}}$ and $w=a^\omega$. Recall that 
$L_{a\text{-seq}}=\infty a \wedge (\infty aa \to \infty aaa)$.
We have that $\wordrank_{L_{a\text{-seq}}}(w)=0$. The sequence $I''=1,2,3,4,...$ cannot be used to show this, since, e.g. $(ab)^\omega\in\inj_{I''}(w)$ yet $\wordrank_{L_{a\text{-seq}}}((ab)^\omega)=2>\wordrank_{L_{a\text{-seq}}}(w)$. 
The sequence $I'=5,10,15,20,...$ however, can be used, as for any injection in these points, the rank of the resulting word will be $0$. Put otherwise, the definition of natural rank~\cite{EhlersS22} is not specific about which sequence of points can be used, it merely requires that such a sequence exists.
\end{remark}
   
\begin{example}[$\infixrank$ for $L_{\text{mod}2}$]\label{ex:infix-rank-mod-2}
    Consider $L_{\text{mod}2}=\infty a~\bigvee~(|w|_{a}$ is even $\wedge~\neg \infty c)~\bigvee~(|w|_{a}$ is odd $\wedge~\neg \infty b)$ over $\Sigma=\{a,b,c\}$
    where $|w|_a$ is the number of occurrences of the letter $a$ in $w$. 
    It has two equivalence classes $e_0$ and $e_1$ where
    for $i\in\{0,1\}$ the equivalence class $e_i$ consists of all word $w$ for which $(|w|_a \mod 2)=i$.
    We have that 
    $$\infixrank_{L_{\text{mod}2}}(u,v)=\left\{ 
\begin{array}{l@{\quad}l}
   2  &  \text{if } u \in e_0 \text{ and } v=b^* \\
   2  &  \text{if } u\in e_1 \text{ and } v=c^* \\
   1  &  \text{if } u\in e_0 \text{ and } v=(b\cup c)^*c(b\cup c)^* \\
   1  &  \text{if } u\in e_1 \text{ and } v=(b\cup c)^*b(b\cup c)^* \\
   0 & \text{if } u\in e_0  \text{ and } v=\Sigma^*a\Sigma^* \\
   0 & \text{if } u\in e_1  \text{ and } v=\Sigma^*a\Sigma^*
\end{array}
    \right.$$
\end{example}

\begin{example}[Computing whether an index is influential]\label{ex:influential}
    Consider the language $L_{a\text{-seq}}$ from \autoref{ex:diff-priorities-for-red-yellow-green-cycles}.
    Let $w=abbaaaaabbba^\omega$. We show how to calculate whether $k$ is influential wrt  $i=1$ for $k=3,5$.
    
    For $k=3$, following \autoref{def:influential}, we check whether there exists $j\in[1..3]$ for which there is a decrease in the natural rank of $w[j..k)$ when the last index $k$ is taken into account. We thus compare the natural rank of $\epsilon$ to the rank of $b$ (for $j{=}3$), of $b$ compared to $bb$ (for $j{=}2)$, and of $ab$ compared to $abb$ (for $j{=}1$). We get that $\infixrank_{L_{a\text{-seq}}}(ab,\epsilon)=3$ and $\infixrank_{L_{a\text{-seq}}}(ab,b)=3$. There is no decrease in the natural rank. Similarly, $\infixrank_{L_{a\text{-seq}}}(a,b)=3$ and $\infixrank_{L_{a\text{-seq}}}(a,bb)=3$ and $\infixrank_{L_{a\text{-seq}}}(\epsilon,ab)=2$ and $\infixrank_{L_{a\text{-seq}}}(\epsilon,abb)=2$.
    That is, there is no decrease in natural rank of $w[j..k)$ for any $j$. Therefore, $k=3$ is not influential wrt $i=1$.

    For $k=5$, for searching for the index $j$ that follows the requirements of \autoref{def:influential} we compare the natural rank of $\epsilon$ and $a$ (for $j{=}5$), of $a$ and $aa$ (for $j{=}4)$, of $ba$ and $baa$ (for $j{=}3)$, of $bba$ and $bbaa$ (for $j{=}2)$, and last of $abba$ and $abbaa$ (for $j{=}1$). We have $\infixrank_{L_{a\text{-seq}}}(abba,\epsilon)=3$ and $\infixrank_{L_{a\text{-seq}}}(abba,a)=2$, showing there is a decrease in the natural rank. 
    Thus, there is at least one index $j$ for which there is a decrease in the natural rank of $w[j..k)$, then $k{=}5$ is influential wrt $i{=}1$.
\end{example}

\begin{example}[Dominant suffixes and natural ranks of letters]\label{ex:nat-rank-calcs}
    Consider $L_{a\text{-seq}}$. Recall that $L_{a\text{-seq}}$ is $\infty a \wedge (\infty aa \to \infty aaa)$.
    \autoref{ex:seq-of-natrank-domnid-inf-a-if-aa-then-aaa} shows the calculation of the natural rank of the first $9$ letters of the word $w{=}abbaaaaab{...}$, where the row $i_w(k)$ signifies whether the $k$-th index is influential. The role of the other rows can be inferred from the row title.
    \autoref{ex:seq-of-natrank-domnid-inf-ab} does the same for the word $w{=}bbaababbb{...}$ and $L_{\infty ab}$.
\end{example}
\noindent
\begin{minipage}{0.5\textwidth}
\setlength{\tabcolsep}{1.2pt}
\begin{tabular}{r@{\ }|@{\ }lllllllllll}
$w$ &   & $a$ & $b$ & $b$ & $a$ & $a$ & $a$ & $a$ & $a$ & $b$ & {...} \\
  $k$   &  \grayit{0} & 1 & 2 & 3 & 4 & 5 & 6 & 7 & 8 & 9 & {...}\\ 
\hline
$i_w(k)$  &   & \true & \false & \false & \true & \true & \true & \true & \true & \false & {...}\\
$\reset_w(k)$ & \grayit{0} & 1 & 2 & 3 & 4 & 4 & 4 & 5 & 6 & 9 & {...}\\
$\predomsuf_{L_{a\text{-seq}}}(w[..k])$  & \rotatebox{90}{$\grayit{\varepsilon}$} & \rotatebox{90}{$a$} & \rotatebox{90}{$b$} & \rotatebox{90}{$bb$} & \rotatebox{90}{$a$} & \rotatebox{90}{$aa$} & \rotatebox{90}{$aaa$} & \rotatebox{90}{$aaa$} & \rotatebox{90}{$aaa$}  & \rotatebox{90}{$b$} &{...}\\
$\domsuf_{L_{a\text{-seq}}}(w[..k])$  & \rotatebox{90}{$\grayit{\varepsilon}$} & \rotatebox{90}{$a$} & \rotatebox{90}{$b$} & \rotatebox{90}{$b$} & \rotatebox{90}{$a$} & \rotatebox{90}{$aa$} & \rotatebox{90}{$aaa$} & \rotatebox{90}{$aaa$} & \rotatebox{90}{$aaa$}  & \rotatebox{90}{$b$} &{...}\\
$\letterrank_{L_{a\text{-seq}}}(w,k)$  & \grayit{3} & 2 & 3 & 3 & 2 & 1 & 0 & 0 & 0 & 3 &{...}\\
\end{tabular}
\vspace{2mm}
\captionof{table}{
For $w=abbaaaaab{...}$ and $L_{a\text{-seq}}$}\label{ex:seq-of-natrank-domnid-inf-a-if-aa-then-aaa}
\end{minipage}
\begin{minipage}{0.5\textwidth}
\setlength{\tabcolsep}{1.2pt}
\begin{tabular}{r@{\ }|@{\ }lllllllllll}
$w$ &   & $b$ & $b$ & $a$ & $a$ & $b$ & $a$ & $b$ & $b$ & $b$ & {...}\\
  $k$   &  \grayit{0} & 1 & 2 & 3 & 4 & 5 & 6 & 7 & 8 & 9 &{...} \\ 
\hline
$i_w(k)$  &   & \false & \false & \true & \false & \true & \true & \true & \false & \false & {...}\\
$\reset_w(k)$  & \grayit{0} & 1 & 2 & 2 & 4 & 4 & 5 & 6 & 8 & 9 &{...}\\
$\predomsuf_{L_{\infty ab}}(w[..k])$  & \rotatebox{90}{$\grayit{\varepsilon}$} & \rotatebox{90}{$b$} & \rotatebox{90}{$bb$} & \rotatebox{90}{$ba$} & \rotatebox{90}{$aa\phantom{a}$} & \rotatebox{90}{$ab$} & \rotatebox{90}{$ba$} & \rotatebox{90}{$ab$} & \rotatebox{90}{$bb$}  & \rotatebox{90}{$bbb$} &{...}\\
$\domsuf_{L_{\infty ab}}(w[..k])$  & \rotatebox{90}{$\grayit{\varepsilon}$} & \rotatebox{90}{$b$} & \rotatebox{90}{$b$} & \rotatebox{90}{$ba$} & \rotatebox{90}{$a\phantom{aa}$} & \rotatebox{90}{$ab$} & \rotatebox{90}{$ba$} & \rotatebox{90}{$ab$} & \rotatebox{90}{$b$}  & \rotatebox{90}{$b$} &{...}\\

$\letterrank_{L_{\infty ab}}(w,k)$  & \grayit{1} & 1 & 1 & 0 & 1 & 0 & 0 & 0 & 1 & 1 & {...}\\
\end{tabular}
\vspace{2mm}
\captionof{table}{
For $w=bbaababbb{...}$ and $L_{\infty ab}$}\label{ex:seq-of-natrank-domnid-inf-ab}
\end{minipage}

\begin{example}[Robustness values]\label{ex:rbstval}
    Recall $L_{\infty a}$ and the words $w_1=bb(ab)^\omega$, $w_2=bbbbb(ab)^\omega$, $w_3=aa(ab)^\omega$ and $w_4=aaaaa(ab)^\omega$. Note that by applying \autoref{def:letter-rank} all $a$'s get rank $0$ and all $b$'s get rank $1$.
    The spoke of $w_1$ thus consists of two \col{red} letters and the period consists of one \col{green} and one \col{yellow} letters. Thus, $\tau_u$ is $(0,-1)$ and $\tau_v$ is $(0,\frac{1}{2})$.
    The robustness value of $w_1$ is $(1,(0,\frac{1}{2}),2\cdot(0,-1\frac{1}{2}))=(1,(0,\frac{1}{2}),(0,-3))$. Similarly, the robustness value of $w_2$ is $(1,(0,\frac{1}{2}),5\cdot(0,-1\frac{1}{2}))=(1,(0,\frac{1}{2}),(0,-7\frac{1}{2}))$. Thus, $w_1$ is more robust than $w_2$. 

    The spoke of $w_3$ consists of two \col{green} letters and the period consists of one \col{green} and one \col{yellow} letters. Thus, $\tau_u$ is $(0,1)$ and $\tau_v$ is $(0,\frac{1}{2})$.
    The robustness value of $w_3$ and $w_4$ are thus $(1,(0,\frac{1}{2}),2\cdot(0,\frac{1}{2}))=(1,(0,\frac{1}{2}),(0,1))$ and $(1,(0,\frac{1}{2}),5\cdot(0,\frac{1}{2}))=(1,(0,\frac{1}{2}),(0,2\frac{1}{2}))$, respectively. Thus, $w_4$ is more robust than $w_3$. 

    Consider $w_5=aa(b)^\omega$ and $w_6=aaaaa(b)^\omega$ which are rejected. The spoke of $w_5$ consists of two \col{green} letters and one \col{yellow} letter while the period consists of only \col{red} letters. Then $\tau_u$ is $(0,1)$ and $\tau_v$ is $(0,-1)$.
    The robustness value of $w_5$ and $w_6$ are thus $(-1,(0,-1),2\cdot(0,2))=(-1,(0,-1),(0,4))$, and $(-1,(0,-1),5\cdot(0,2))=(-1,(0,-1),(0,10))$, respectively. Hence, $w_6$ is more robust than $w_5$. 
\end{example}

\begin{example}[The precise DPA for $L_{a\text{-seq}}$]\label{ex:precise-aaa}
    Consider the language $L_{a\text{-seq}}$. The precise DPA~\cite{BohnL24} for $L_{a\text{-seq}}$ is $\aut{P}_{a\text{-seq}}^{\text{Precise}}$ from \autoref{fig:precise}.
    The precise DPA is transition-based rather than state-based (i.e. the ranks are on the transitions rather than the {states}) but this is not the issue. 
    The issue is that it does not rank $(a)^\omega$ with $210000...$ as does the sequence of $n_1,n_2,...$ where $n_i=\letterrank_{L_{a\text{-seq}}}(a^\omega,i)$.
    Consequently it suggests that $a^\omega$ is not the most robust word for $L_{a\text{-seq}}$, contrary to what we require. Indeed, according to it $(ab)^\omega$ gets a better value (if we had used a reasonable adjustment of the coloring/score from state-based to transition-based).
\end{example}

\commentout{
\begin{example}[Four different DPAs for $L_{a\text{-seq}}$]\label{ex:four-dpa-aaa}
Continuing the example with regard to 
the language $L_{a\text{-seq}}$, 
we note that the DPA $\aut{P}_3$ of \autoref{fig:aut-for-examples-parity} recognizes $L_{a\text{-seq}}$ and provides the desired correlation between the natural ranks of letters and state ranks. Note that by redirecting the self-loop on $a$ from the state ranked $0$ to any other state, we do not change the language. Thus, we get four different state-based DPAs while only one of which is the desired one (has the tight correlation to letter-ranks).
\end{example}
}

\begin{figure} 
\begin{center}
\scalebox{0.65}{
\begin{tikzpicture}[->,>=stealth',shorten >=1pt,auto,node distance=2.0cm,semithick,initial text=, initial above]

\coordinate (P20) at (1.25,0);
\draw[rotate=0,dashed,green,line width=0.5mm] (P20) ellipse (2.6cm and 0.9cm);
\coordinate (P21) at (2.5,0);
\draw[rotate=0,dashed,red,line width=0.5mm] (P21) ellipse (1.1cm and 0.65cm);

\node[state,initial]    (q0)               {$0$};
\node[state]    (q1)  [right of=q0]   {$1$};
\node[label] (qL) [above left of=q0, node distance=1.6cm] {$\aut{P}_1:$};

\path (q0) edge [loop left] 
           node   {$a$}
      (q0); 
\path (q0) edge [bend left] 
           node {$b$} 
      (q1);
\path (q1) edge [bend left] 
           node {$a$} 
      (q0);
\path (q1) edge [loop right] 
           node  {$b$}
      (q1);

\coordinate (P10) at (6,-0.4);
 \draw[rotate=90,dashed,green,line width=0.5mm] (P10) ellipse (1.1cm and 0.65cm);
 \coordinate (P11) at (8,-0.4);
 \draw[rotate=90,dashed,red,line width=0.5mm] (P11) ellipse (1.1cm and 0.65cm);    

\node[state,initial]    (r0)   [right of=q1, node distance=4cm]            {$0$};
\node[state]    (r1)  [right of=r0]   {$-1$};
\node[label] (qL) [above left of=r0, node distance=1.6cm] {$\aut{P}_2:$};

\path (r0) edge [loop below] 
           node   {$a$}
      (r0); 
\path (r0) edge  
           node {$b$} 
      (r1);
\path (r1) edge [loop below] 
           node  {$a,b$}
      (r1);

\coordinate (P30) at (14.75,0);
\draw[rotate=0,dashed,green,line width=0.5mm] (P30) ellipse (4.6cm and 1.5cm);
\coordinate (P31) at (13.5,0);
\draw[rotate=0,dashed,red,line width=0.5mm] (P31) ellipse (3.11cm and 1.15cm);
\coordinate (P32) at (12.58,0);
\draw[rotate=0,dashed,green,line width=0.5mm] (P32) ellipse (2.05cm and 0.9cm);
\coordinate (P33) at (11.75,0);
\draw[rotate=0,dashed,red,line width=0.5mm] (P33) ellipse (1.01cm and 0.65cm); 

\node[state,initial]  [right of=r1, node distance=4cm] (p3)          {$3$};
\node[state]  (p2)  [right of=p3]   {$2$};
\node[state]  (p1)  [right of=p2]   {$1$};
\node[state]  (p0)  [right of=p1]   {$0$};
\node[label] (pL) [above left of=p3, node distance=1.6cm] {$\aut{P}_3:$};

\path (p3) edge   
           node [below] {$a$}
      (p2); 
\path (p2) edge   
           node [below] {$a$}
      (p1);
\path (p1) edge   
           node [below] {$a$}
      (p0);
\path (p0) edge  [loop right] 
           node [right] {$a$}
      (p0);      
      
\path (p3) edge [loop left]  
           node [left] {$b$}
      (p3);   
\path (p2) edge [bend right=25]  
           node [above, near start]  {$b$}
      (p3);   
\path (p1) edge [bend right=35]  
           node [above, near start] {$b$}
      (p3);   
\path (p0) edge [bend right=45]  
           node [above, near start] {$b$}
      (p3);

\end{tikzpicture}}
\end{center}
\vspace{-6mm}
\caption{
Parity automata $\aut{P}_1,\aut{P}_2,\aut{P}_3$ for the languages $L_{\infty a}$, $L_{\ltlG a}$, and $L_{a\text{-seq}}$, respectively, annotated with accepting (dashed green) and rejecting SCCs (dashed red) witnessing their inclusion measure.
Recall that 
$L_{a\text{-seq}}$ is $\infty a \wedge (\infty aa \to \infty aaa)$.
}\label{fig:aut-for-examples-parity}
\end{figure}

\begin{figure}
\begin{center}
\scalebox{0.65}{
\begin{tikzpicture}[->,>=stealth',shorten >=1pt,auto,node distance=2.0cm,semithick,initial text=, initial above]

\node[state]    (c0)               {};
\node[state]    (c1)  [right of=c0]   {};
\node[state]    (c2)  [right of=c1]   {};

\node[label] (qcL) [above left of=c0, node distance=1.25cm] {$\aut{P}_{a\text{-seq}}^{\text{Precise}}:$};

\path (c0) edge [loop left] 
           node   {$b/3$}
      (c0); 
\path (c0) edge 
           node   {$a/2$}
      (c1); 
\path (c1) edge 
           node   {$a/1$}
      (c2); 
\path (c2) edge   [bend right=45]
           node   {$a/0$}
      (c0); 
\path (c2) edge   [bend left=50]
           node   {$b/1$}
      (c0); 
\path (c1) edge   [bend left=30]
           node   {$b/2$}
      (c0); 
%%%%%%%%%%%%%%%%%%%%%%%%%%%%%%%%%%%%%%%%%%%

\node[state]    (d0)   [right of=c2, node distance=6cm]            {};
\node[state]    (d1)  [right of=d0]   {};
\node[state]    (d2)  [right of=d1]   {};

\node[label] (qdL) [above left of=d0, node distance=1.25cm] {$\aut{P}_{\infty ab}^{\text{Precise}}:$};

\path (d0) edge [loop left] 
           node   {$b/1$}
      (d0); 
\path (d0) edge [bend left] 
           node   {$a/0$}
      (d1); 
\path (d1) edge [bend left] 
           node   {$a/1$}
      (d2); 
\path (d2) edge [bend left] 
           node   {$b/0$}
      (d1); 
\path (d1) edge [bend left] 
           node   {$b/1$}
      (d0); 
\path (d2) edge [loop right] 
           node   {$a/1$}
      (d2); 
      
\end{tikzpicture}}
\end{center}
\vspace{-6mm}
\caption{The precise parity automata $\aut{P}_{a\text{-seq}}^{\text{Precise}}$ and $\aut{P}_{\infty ab}^{\text{Precise}}$ for $L_{a\text{-seq}}$ and $L_{\infty ab}$ resp. over $\Sigma=\{a,b\}$.}\label{fig:precise}
\end{figure}
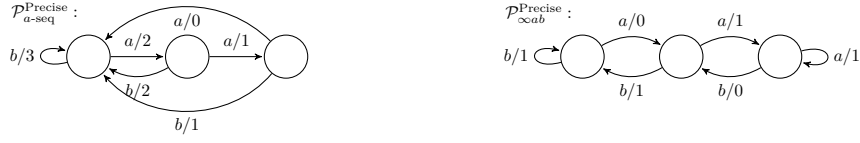

\begin{example}[The precise DPA for $L_{\infty ab}$] \label{ex:precise-ab}
    Consider again the language $L_{\infty ab}$ from
    \autoref{fig:dpa-infty-ab}. The precise DPA~\cite{BohnL24} for $L_{\infty ab}$ is $\aut{P}_{\infty ab}^{\text{Precise}}$ from \autoref{fig:precise}. We can see that the word $(ab)^\omega$ is ranked with $(10)^\omega$ instead of $10^\omega$ as correlated with the natural rank and as we expect from the robustness DPA (see \autoref{ex:infty-ab}).
\end{example}

\begin{figure}
\begin{center}
\scalebox{0.62}{
\begin{tikzpicture}[->,>=stealth',shorten >=1pt,auto,node distance=2.0cm,semithick,initial text=, initial left]

\node[state,rectangle, rounded corners=0.35cm,initial]    (a0)
{$\begin{array}{c}L_0 \\ \aut{P}_{0\epsilon}\ \ \aut{P}_{1\epsilon} \\ \ \colorbox{lipicsLightGray}{2} \ \ \ \ \ 2\end{array}$};

\node[state,rectangle, rounded corners=0.35cm]    (a1)  [above right of=a0, node distance=3cm]
{$\begin{array}{c}L_0 \\ \aut{P}_{0c}\ \ \aut{P}_{1\epsilon} \\ \ \colorbox{lipicsLightGray}{1} \ \ \ \ \ 2\end{array}$};

\node[state,rectangle, rounded corners=.35cm]    (a2)  [above left of=a0, node distance=3cm]
{$\begin{array}{c}L_0 \\ \aut{P}_{0\epsilon}\ \ \aut{P}_{1b} \\ \ \colorbox{lipicsLightGray}{2} \ \ \ \ \ 1\end{array}$};

\node[state,rectangle, rounded corners=.35cm]    (a3)  [below of=a0, node distance=3cm]
{$\begin{array}{c}L_1 \\ \aut{P}_{0a}\ \ \aut{P}_{1a} \\ \ 0 \ \ \ \ \ \colorbox{lipicsLightGray}{0}\end{array}$};

\node[state,rectangle, rounded corners=.35cm]    (a4)  [below of=a3, node distance=3cm]
{$\begin{array}{c}L_0 \\ \aut{P}_{0a}\ \ \aut{P}_{1a} \\ \ \colorbox{lipicsLightGray}{0} \ \ \ \ \ 0\end{array}$};

\node[state,rectangle, rounded corners=.35cm]    (a5)  [below left of=a4, node distance=3cm]
{$\begin{array}{c}L_1 \\ \aut{P}_{0\epsilon}\ \ \aut{P}_{1b} \\ \ 2 \ \ \ \ \ \colorbox{lipicsLightGray}{1}\end{array}$};

\node[state,rectangle, rounded corners=.35cm]    (a6)  [below right of=a4, node distance=3cm]
{$\begin{array}{c}L_1 \\ \aut{P}_{0c}\ \ \aut{P}_{1\epsilon} \\ \ 1 \ \ \ \ \ \colorbox{lipicsLightGray}{2}\end{array}$};

\node[label] (qaL) [above left of=a2, node distance=1.5cm] {$\aut{P}^{\vigor}_{L_{\text{mod}2}}:$};

\path (a0) edge 
           node [left,near start]  {$b$}
      (a2); 
\path (a0) edge 
           node   {$a$}
      (a3); 
\path (a0) edge 
           node  [right,near start] {$c$}
      (a1); 

\path (a1) edge 
           node [right]  {$a$}
      (a3); 
\path (a1) edge 
           node   {$b$}
      (a2); 
\path (a1) edge  [loop right]
           node   {$c$}
      (a1); 

\path (a2) edge 
           node [left]  {$a$}
      (a3); 
\path (a2) edge  [loop left]
           node   {$b$}
      (a2); 
\path (a2) edge  [bend left]
           node   {$c$}
      (a1); 

\path (a3) edge [bend left]
           node   {$a$}
      (a4); 
\path (a3) edge  [bend right]
           node   {$b$}
      (a5); 
\path (a3) edge  [bend left]
           node   {$c$}
      (a6); 

\path (a4) edge [bend left]
           node   {$a$}
      (a3); 
\path (a4) edge  [bend left]
           node   {$b$}
      (a2); 
\path (a4) edge  [bend right]
           node   {$c$}
      (a1); 

\path (a5) edge 
           node   {$a$}
      (a4); 
\path (a5) edge [loop left]
           node   {$b$}
      (a5); 
\path (a5) edge  [bend right]
           node   {$c$}
      (a6); 

\path (a6) edge 
           node   {$a$}
      (a4); 
\path (a6) edge 
           node   {$b$}
      (a5); 
\path (a6) edge  [loop right]
           node   {$c$}
      (a6); 

%%%%%%%%%%%%%%%%%%%%%%%%%%%%%%%%%%%%%%%%%%%%

\node[state,initial, right of=a0, node distance=8.25cm]    (b0)
{$2$};

\node[state]    (b1)  [above right of=b0, node distance=3cm]
{$1$};

\node[state]    (b2)  [above left of=b0, node distance=3cm]
{$2$};

\node[state]    (b3)  [below of=b0, node distance=3cm]
{$0$};

\node[state]    (b4)  [below of=b3, node distance=3cm]
{$0$};

\node[state]    (b5)  [below left of=b4, node distance=3cm]
{$1$};

\node[state]    (b6)  [below right of=b4, node distance=3cm]
{$2$};

\node[label] (qbL) [above left of=b2, node distance=1.5cm] {$\aut{P}^\vigor_{L_{\text{mod}2}}:$};

\path (b0) edge 
           node [left]  {$b$}
      (b2); 
\path (b0) edge 
           node   {$a$}
      (b3); 
\path (b0) edge 
           node [right]  {$c$}
      (b1); 

\path (b1) edge 
           node [right]  {$a$}
      (b3); 
\path (b1) edge 
           node   {$b$}
      (b2); 
\path (b1) edge  [loop right]
           node   {$c$}
      (b1); 

\path (b2) edge 
           node [left]  {$a$}
      (b3); 
\path (b2) edge  [loop left]
           node   {$b$}
      (b2); 
\path (b2) edge  [bend left]
           node   {$c$}
      (b1); 

\path (b3) edge [bend left]
           node   {$a$}
      (b4); 
\path (b3) edge  [bend right]
           node   {$b$}
      (b5); 
\path (b3) edge  [bend left]
           node   {$c$}
      (b6); 

\path (b4) edge [bend left]
           node   {$a$}
      (b3); 
\path (b4) edge  [bend left]
           node   {$b$}
      (b2); 
\path (b4) edge  [bend right]
           node   {$c$}
      (b1); 

\path (b5) edge 
           node   {$a$}
      (b4); 
\path (b5) edge [loop left]
           node   {$b$}
      (b5); 
\path (b5) edge  [bend right]
           node   {$c$}
      (b6); 

\path (b6) edge 
           node   {$a$}
      (b4); 
\path (b6) edge 
           node   {$b$}
      (b5); 
\path (b6) edge  [loop right]
           node   {$c$}
      (b6); 

%%%%%%%%%%%%%%%%%%%%%%%%%%%%%%%

\node[state,initial]  (l0)  [left of=b2, node distance=15.25cm] {$L_0$};
\node[state]    (l1)  [right of=l0, node distance=3.00cm]   {$L_1$};

\node[label] (qlL) [above left of=l0, node distance=1.25cm] {$\aut{A}[\sim_{L_{\text{mod}2}}]:$};

\path (l0) edge [loop below] 
           node   {$b,c$}
      (l0); 
\path (l1) edge [loop below] 
           node   {$b,c$}
      (l1); 

\path (l0) edge [bend left] 
           node {$a$} 
      (l1);
\path (l1) edge [bend left] 
           node {$a$} 
      (l0);

%%%%%%%%%%%%%%%%%%%%%%%%%%%%%%%%%%%%%

\node[label]    (qi)  [left of=b4,node distance=14.5cm]   {};
\node[state]    (q1)  [below of=qi,node distance=2cm]   {$1,b$};
\node[state,initial]  (q0)  [above left of=q1] {$2,\epsilon$};
\node[state]    (q2)  [below left of=q0]   {$0,a$};

\node[label] (qqL) [above left of=q0, node distance=1.25cm] {$\aut{P}_{a}^\vigor:$};

\path (q0) edge [loop above] 
           node   {$c$}
      (q0); 
\path (q2) edge [loop below] 
           node   {$a$}
      (q2); 
\path (q1) edge [loop below] 
           node   {$b$}
      (q1); 
\path (q0) edge  
           node {$a$} 
      (q2);
\path (q0) edge [bend left] 
           node {$b$} 
      (q1);
\path (q1) edge [bend left] 
           node {$a$} 
      (q2);
\path (q1) edge 
           node {$c$} 
      (q0);
\path (q2) edge [bend left] 
           node [near start] {$c$} 
      (q0);
\path (q2) edge 
           node {$b$} 
      (q1);

%%%%%%%%%%%%%%%%%%%%%%%%%%%%%%%%%%%%%%%

\node[state,initial]    (r0)    [above of=q0, node distance=4.75cm]           {$2,\epsilon$};
\node[state]    (r1)  [below right of=r0]   {$1,c$};
\node[state]    (r2)  [below left of=r0]   {$0,a$};

\node[label] (qL) [above left of=r0, node distance=1.25cm] {$\aut{P}_{\varepsilon}^\vigor:$};

\path (r0) edge [loop above] 
           node   {$b$}
      (r0); 
\path (r2) edge [loop below] 
           node   {$a$}
      (r2); 
\path (r1) edge [loop below] 
           node   {$c$}
      (r1); 
\path (r0) edge  
           node {$a$} 
      (r2);
\path (r0) edge [bend left] 
           node {$c$} 
      (r1);
\path (r1) edge [bend left] 
           node {$a$} 
      (r2);
\path (r1) edge 
           node {$b$} 
      (r0);
\path (r2) edge [bend left] 
           node [near start] {$b$} 
      (r0);
\path (r2) edge 
           node {$c$} 
      (r1);

\end{tikzpicture}}
\end{center}
\vspace{-2mm}
\caption{
The vigor DPA for $L_{\text{mod}2}=\infty a~\bigvee~(|w|_{a}$ is even $\wedge~\neg \infty c)~\bigvee~(|w|_{a}$ is odd $\wedge~\neg \infty b)$ over $\Sigma=\{a,b,c\}$, and the related components:
$\aut{A}[\sim_{L_{\text{mod}2}}]$ induced by the equivalence classes of $L_{\text{mod}2}$, parity automata $\aut{P}{\varepsilon}^\vigor$ and $\aut{P}_{a}^\vigor$ for the respective languages, as per \autoref{def:rbst-dpa}.
There are two versions of the vigor DPA $\aut{P}^{\vigor}_{L_{\text{mod}2}}$ for $L_{\text{mod}2}$ : the one on the left shows the relation to the components as per \autoref{def:rbst-dpa}, whereas the one on the right shows only the final ranks.\label{fig:P-rbst}}
\end{figure}
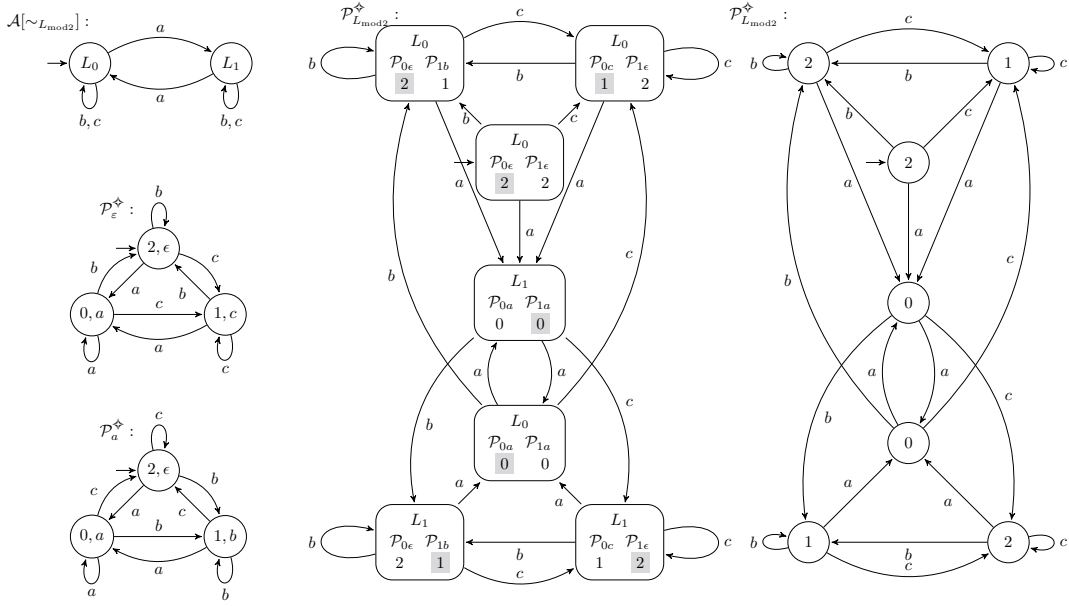

\begin{example}[Black edges issue]\label{ex:black-edges-issue}
Consider the DPA $\aut{P}_2$ from~\autoref{fig:aut-for-examples-parity}. Applying the edge-coloring results in $\aut{P}'_2$ of \autoref{fig:aut-for-examples-parity-colored}, so after reading a black edge all edges will be black. Thus the colors traversed when reading e.g. $aaabaaa$ will be \col{g}\col{g}\col{g}\col{b}\col{b}\col{b}\col{b} where \col{g} stands for \col{green} and \col{b} for \col{black} whereas we want \col{g}\col{g}\col{g}\col{b}\col{g}\col{g}\col{g}. This mismatch will be encountered whenever a black or a white letter exists. This is since a black (resp. white) edge will transition to a rejecting (resp. accepting) MSCC, from which all edges will be black (resp. white). However, in \autoref{def:letter-color} after a black or a white letter is observed, other colors can be observed since the definition makes use of the forgetful version of the word which essentially forgets the black and white letters seen before. 
This issue 
called for the definition of forgetful-parity acceptance condition, from which we derive the  robustness DPA $\aut{P}_{L_{\ltlG a}}^\truerobustness$ shown in~\autoref{fig:aut-for-examples-parity-colored}.
\end{example}

\begin{figure}
\begin{center}
\scalebox{0.65}{
\begin{tikzpicture}[->,>=stealth',shorten >=1pt,auto,node distance=2.0cm,semithick,initial text=, initial below]

\node[state,initial]    (r0)             {$0$};
\node[state]    (r1)  [right of=r0]   {$-1$};
\node[label] (qL) [above left of=r0, node distance=1.6cm] {$\aut{P}_2':$};

\path (r0) edge [pgreen, loop left] 
           node   {$a$}
      (r0); 
\path (r0) edge  [pblack]
           node {$b$} 
      (r1);
\path (r1) edge [pblack, loop right] 
           node  {$a,b$}
      (r1);          

\node[state,initial]    (k0)   [below of=r0, node distance=2.5cm]            {$0$};
\node[state]    (k1)  [right of=k0]   {$-1$};
\node[label] (qkL) [above left of=k0, node distance=1.6cm] {$\aut{P}_{L_{\ltlG a}}^\truerobustness:$};

\path (k0) edge [pgreen, loop left] 
           node   {$a$}
      (k0); 
\path (k0) edge [pblack, bend left] 
           node {$b$} 
      (k1);
\path (k1) edge [dotted, bend left] 
           node {$\epsilon$} 
      (k0);

\node[state]    (q1)  [left of=k0, node distance=4cm]   {$1$};

\node[state,initial]    (q0)     [left of=q1]     {$0$};
\node[label] (qL) [above left of=q0, node distance=1.6cm] {$\aut{P}_{L_{\infty a}}^\truerobustness:$};

\path (q0) edge [pgreen, loop above] 
           node   {$a$}
      (q0); 
\path (q0) edge [pyellow, bend left] 
           node {$b$} 
      (q1);
\path (q1) edge [pgreen, bend left] 
           node {$a$} 
      (q0);
\path (q1) edge [pred, loop above] 
           node  {$b$}
      (q1);

\node[state,initial]  [right of=k1, node distance=4cm] (p3)          {$3$};
\node[state]  (p2)  [right of=p3]   {$2$};
\node[state]  (p1)  [right of=p2]   {$1$};
\node[state]  (p0)  [right of=p1]   {$0$};
\node[label] (pL) [above left of=p3, node distance=1.6cm] {$\aut{P}_{L_{a\text{-seq}}}^\truerobustness:$};

\path (p3) edge  [pgreen] 
           node [below] {$a$}
      (p2); 
\path (p2) edge  [pred] 
           node [below] {$a$}
      (p1);
\path (p1) edge  [pgreen] 
           node [below] {$a$}
      (p0);
\path (p0) edge  [pgreen, loop right] 
           node [right] {$a$}
      (p0);      
      
\path (p3) edge [pred, loop above]  
           node  {$b$}
      (p3);   
\path (p2) edge [pyellow, bend right=25]  
           node [above, near start]  {$b$}
      (p3);   
\path (p1) edge [pyellow, bend right=35]  
           node [above, near start] {$b$}
      (p3);   
\path (p0) edge [pyellow, bend right=45]  
           node [above, near start] {$b$}
      (p3);

\end{tikzpicture}}
\end{center}
\vspace{-6mm}
\caption{The robustness parity automata $\aut{P}_{L_{\infty a}}^\truerobustness,\aut{P}_{L_{\ltlG a}}^\truerobustness,\aut{P}_{L_{a\text{-seq}}}^\truerobustness$ for languages $L_{\infty a}$, $L_{\ltlG a}$, $L_{a\text{-seq}}$, resp. The DPA $\aut{P}'_2$ is what we would get
from $\aut{P}_2$ of \autoref{fig:aut-for-examples-parity} if we hadn't apply the conversion to the forgetful parity automaton as per \autoref{prop:edge-color-correctness}, which is not what we want as explained in  \autoref{ex:black-edges-issue}.}\label{fig:aut-for-examples-parity-colored}
\end{figure}

\begin{figure}
\begin{center}
\scalebox{0.65}{
\begin{tikzpicture}[->,>=stealth',shorten >=1pt,auto,node distance=2.0cm,semithick,initial text=, initial above]

\node[state]    (r0)               {$-1$};
\node[state,initial]    (r1)  [right of=r0]   {$0$};
\node[state]    (r2)  [right of=r1]   {$1$};

\node[label] (qL) [above left of=r0, node distance=1.25cm] {$\aut{P}_{L''}^\vigor:$};

\path (r0) edge [loop left] 
           node   {$\Sigma$}
      (r0); 
\path (r1) edge  
           node {$rg,\overline{r}g$} 
      (r0);
\path (r1) edge [bend left] 
           node {$r\overline{g}$} 
      (r2);
\path (r2) edge [bend left] 
           node {$\overline{r}g$} 
      (r1);
\path (r1) edge [loop below] 
           node  {$\overline{rg}$}
      (r1);    
\path (r2) edge [loop right] 
           node  {$\Sigma\setminus\{\overline{r}g\}$}
      (r2);     

%%%%%%%%%%%%%%%%%%%%%%%%%%%%%%%%%%%%%%%%%%%%%%%%%%%%%%%%%%

\node[state]  [right of=r2, node distance=4cm]  (q0)       {$-1$};
\node[state,initial]    (q1)  [right of=q0]   {$0$};
\node[state]    (q2)  [right of=q1]   {$1$};

\node[label] (qL) [above left of=q0, node distance=1.25cm] {$\aut{P}_{L''}^\truerobustness:$};

\path (q1) edge  [pblack, bend left]
           node {$rg,\overline{r}g$} 
      (q0);
\path (q0) edge  [bend left, dashed]
           node {$\varepsilon$} 
      (q1);
\path (q1) edge [pyellow, bend left] 
           node {$r\overline{g}$} 
      (q2);
\path (q2) edge [pgreen, bend left] 
           node {$\overline{r}g$} 
      (q1);
\path (q1) edge [pgreen, loop below] 
           node  {$\overline{rg}$}
      (q1);    
\path (q2) edge [pred, loop right] 
           node  {$\Sigma\setminus\{\overline{r}g\}$}
      (q2);     
      
\end{tikzpicture}}
\end{center}
\vspace{-6mm}
\caption{The vigor and robustness DPAs  $\aut{P}_{L''}^\vigor$ and $\aut{P}_{L''}^\truerobustness$ for \autoref{ex:no-redundant-grants} with $\Sigma=\{rg,r\overline{g},\overline{r}g,\overline{rg}\}$.}\label{fig:no-redundant-grants}
\end{figure}

\begin{figure}[h]
\begin{center}
\scalebox{0.65}{
\begin{tikzpicture}[->,>=stealth',shorten >=1pt,auto,node distance=2.0cm,semithick,initial text=, initial above]

\node[state,initial]    (r0)               {$0$};
\node[state]    (r1)  [right of=r0]   {$1$};
\node[state]    (r2)  [right of=r1]   {$1$};
\node[state]    (r3)  [below of=r1]   {$-1$};

\node[label] (qL) [above left of=r0, node distance=1.25cm] {$\aut{P}_{L'''}^\vigor:$};

\path (r0) edge [loop left] 
           node   {$s$}
      (r0); 
\path (r0) edge [bend left=55] 
           node {$d$} 
      (r2);
\path (r1) edge  
           node {$s$} 
      (r0);
\path (r2) edge 
           node {$r$} 
      (r1);
\path (r1) edge [loop above] 
           node [left] {$r$}
      (r1);    
\path (r2) edge [loop right] 
           node  {$d$}
      (r2);     
\path (r0) edge  
           node {$r$} 
      (r3);
\path (r1) edge  
           node {$d$} 
      (r3);
\path (r2) edge  
           node {$s$} 
      (r3);
\path (r3) edge [loop below] 
           node   {$s,d,r$}
      (r3); 

%%%%%%%%%%%%%%%%%%%%%%%%%%%%%%%%%%%%%%%%%%%%%%%%%%%%%

\node[state,initial]  [right of=r2, node distance=4cm]  (q0)       {$0$};
\node[state]    (q1)  [right of=q0]   {$1$};
\node[state]    (q2)  [right of=q1]   {$1$};
\node[state]    (q3)  [below of=q0]   {$-1$};
\node[state]    (q4)  [right of=q3]   {$-1$};
\node[state]    (q5)  [right of=q4]   {$-1$};

\node[label] (qL) [above left of=q0, node distance=1.25cm] {$\aut{P}_{L'''}^\truerobustness:$};

\path (q0) edge [pgreen, loop left] 
           node   {$s$}
      (q0); 
\path (q0) edge [pyellow, bend left=55] 
           node {$d$} 
      (q2);
\path (q1) edge  [pgreen]
           node {$s$} 
      (q0);
\path (q2) edge [pred]
           node {$r$} 
      (q1);
\path (q1) edge [pred, loop above] 
           node [left] {$r$}
      (q1);    
\path (q2) edge [pred, loop right] 
           node  {$d$}
      (q2);     
\path (q0) edge  [pblack, bend left]
           node {$r$} 
      (q3);
\path (q1) edge  [pblack, bend left]
           node {$d$} 
      (q4);
\path (q2) edge [pblack, bend left] 
           node {$s$} 
      (q5);
\path (q3) edge  [bend left, dashed]
           node {$\varepsilon$} 
      (q0);
\path (q4) edge  [bend left, dashed]
           node {$\varepsilon$} 
      (q1);
\path (q5) edge [bend left, dashed] 
           node {$\varepsilon$} 
      (q2);
      
\end{tikzpicture}}
\end{center}
\vspace{-6mm}
\caption{The vigor parity automaton $\aut{P}_{L'''}^\vigor$ and its forgetful version, the robustness automaton $\aut{P}_{L'''}^\truerobustness$ for \autoref{ex:danger}.}\label{fig:danger}
\end{figure}

\begin{example}[Safe, danger, recover]\label{ex:danger}
    Consider the following property 
    over the alphabet  $\{s,d,r\}$ standing for \signal{safe}, \signal{danger} and \signal{recover}, resp. 
    As long as there are no issues the system is in $\signal{safe}$ mode. Once an issue has occurred, the system moves to $\signal{danger}$ mode.
    When this occurs, the system should eventually move to $\signal{recover}$ mode and until then stay in $\signal{danger}$ mode. From $\signal{recover}$ mode it should eventually move back to $\signal{safe}$ mode.    
    Formally, the property is $L'''=\ltlG(s \to \ltlX [s\ltlW d]) \wedge \ltlG(d \to \ltlX [d \ltlU r]) \wedge \ltlG(r \to \ltlX [r \ltlU s])$. 
    Consider two words $(d^{i}r^{j}s^{k})^\omega$ and $(d^{i'}r^{j'}s^{k'})^\omega$ and assume the period lengths are equal, that is, $i+j+k=i'+j'+k'$. To capture the intuition that safe actions are the preferred ones, we require $(d^{i}r^{j}s^{k})^\omega \gtrbst_{L'''} (d^{i'}r^{j'}s^{k'})^\omega$ if $i+j<i'+j'$, or equivalently $k>k'$.
    
    Its vigor and robustbess DPAs are given in \autoref{fig:danger}.
\end{example}

\end{document}